%% file: z_mainFusTensor5_JB.tex
\documentclass[journal]{IEEEtran}
\usepackage{pgfplotstable}

\input{MathSymbolDefs.tex}

\newcommand{\simZPsi}{\sim_{{Z,\Psi}}}

\definecolor{orange}{rgb}{0., 0., 0.}
\definecolor{blue}{rgb}{0., 0., 0.}
\definecolor{darkgreen}{rgb}{0., 0., 0.}
\definecolor{red}{rgb}{0., 0., 0.}
\definecolor{magenta}{rgb}{0., 0., 0.}

\usepackage[bibbreaks=tight,
            paragraphs=tight,
            floats=normal,
            mathspacing=normal,
            wordspacing=normal,
            tracking=normal,
            bibnotes=tight,
            charwidths=normal,
            mathdisplays=normal,
            leading=normal,
            indent=normal,
            lists=normal,
            bibliography=normal,
            title=normal,
            sections=normal,
            margins=normal,
            ]{savetrees}

\title{Coupled Tensor Decomposition for \\ Hyperspectral and Multispectral Image \\ Fusion with \corange{Inter-Image} Variability}

\author{
Ricardo~A.~Borsoi, Clémence~Prévost, Konstantin~Usevich, David~Brie, Jos\'e~C.~M.~Bermudez, Cédric~Richard
\thanks{This work was  partly  supported  by  the  ANR (Agence Nationale de Recherche) under  grant LeaFleT (ANR-19-CE23-0021), by the National Council for Scientific and Technological Development (CNPq) under grants 304250/2017-1, 409044/2018-0, 141271/2017-5 and 204991/2018-8, by the Foundation for Research Support of the State of Rio Grande do Sul (FAPERGS) under grant 19/2551-0001844-4.
}%
\thanks{We also thank the GdR CNRS ISIS for supporting the collaboration between R.A. Borsoi and C. Pr{\'e}vost, K. Usevich and D. Brie through a \textit{bourse de mobilité}.}
\thanks{R.A. Borsoi is with the Department of Electrical Engineering, Federal University of Santa Catarina (DEE--UFSC), Florian\'opolis, SC, Brazil, and with the Lagrange Laboratory, Universit\'e  C\^ote  d'Azur, Nice, France. e-mail: \mbox{raborsoi@gmail.com}.}
\thanks{C. Pr{\'e}vost, K. Usevich and D. Brie are with Centre de Recherche en Automatique de Nancy (CRAN), Universit{\'e} de Lorraine, CNRS, Vandoeuvre-l{\`e}s-Nancy, France. e-mail: \mbox{firstname.lastname@univ-lorraine.fr}.}
\thanks{J.C.M. Bermudez is with the DEE--UFSC, Florian\'opolis, SC, Brazil. e-mail: \mbox{j.bermudez@ieee.org}.}
\thanks{C. Richard is with the  Universit\'e  C\^ote  d'Azur,  Nice, France (e-mail: cedric.richard@unice.fr), Lagrange Laboratory (CNRS, OCA).}
}

\begin{document}

\maketitle



\begin{abstract}

Coupled tensor approximation has recently emerged as a promising approach for the fusion of hyperspectral and multispectral images, reconciling state of the art performance with strong theoretical guarantees. However, tensor-based approaches previously proposed assume that the different observed images are acquired under exactly the same conditions. A recent work proposed to accommodate \corange{inter-image} spectral variability in the image fusion problem using a matrix factorization-based formulation, but did not account for spatially-localized variations. Moreover, it lacks theoretical guarantees and has a high associated computational complexity.
In this paper, we consider the image fusion problem while accounting for both spatially and spectrally localized changes in an additive model. We first study how the general identifiability of the model is impacted by the presence of such changes. Then, assuming that the high-resolution image and the variation factors admit a Tucker decomposition, two new algorithms are proposed -- one purely algebraic, and another based on an optimization procedure. Theoretical guarantees for the exact recovery of the high-resolution image are provided for both algorithms. Experimental results show that the proposed method outperforms state-of-the-art methods in the presence of spectral and spatial variations between the images, at a smaller computational cost.
\end{abstract}


\section{Introduction}

Hyperspectral (HS) cameras are able to acquire images with very high spectral resolution. However, the fundamental compromise between signal-to-noise ratio, spatial resolution, and spectral resolution means that their spatial resolution is usually low~\cite{shaw2003reviewSpecImaging}. Multispectral (MS) devices, on the other hand, are able to achieve a much higher spatial resolution since they contain only a small number of spectral bands.
An approach that attempts to circumvent the physical limitations of imaging sensors consists in combining HS and MS images (MSI) of the same scene to obtain images with high spatial and spectral resolution, in a multimodal image fusion problem~\cite{yokoya2017HS_MS_fusinoRev}, commonly referred to as \emph{hyperspectral super resolution}.

Different algorithms have been proposed to solve this problem. Early approaches were based on component substitution or on multiresolution analysis, which attempt to extract high-frequency spatial details from the MSI and combine them with the HS image (HSI)~\cite{carper1990componentSubstitutionPansharpening,liu2000MRA_HS_MS_fusion}.
Subspace-based formulations have later received a significant amount of interest as they explore the natural representation of the pixels in an HSI as the linear combination of a small number of spectral signatures~\cite{yokoya2017HS_MS_fusinoRev,Keshava:2002p5667,borsoi2018superpixels1_sparseU}. 
Different algorithms have been proposed following this approach using, e.g., Bayesian formulations~\cite{hardie2004hs_ms_fusionMAP} or sparse representations on learned dictionaries~\cite{wei2015hs_ms_fusionBayesianSparse}, and different kinds of matrix factorization formulations employing sparse and spatial regularizations~\cite{kawakami2011hs_ms_fusionMFsparse,simoes2015HySure}, or estimating both the basis vectors and their coefficients blindly/unsupervisedly from the images~\cite{yokoya2012coupledNMF}.

The natural representation of HSIs and MSIs as 3-dimensional tensors has been successfully exploited for hyperspectral unmixing, denoising~\cite{imbiriba2018ULTRA_V,veganzones2015tensorDecompMultisngularHSIs,imbiriba2018_ULTRA,liu2012denoisingHSIsTensorCPD} and super-resolution. Superior super-resolution performance and exact recovery guarantees have been obtained using this formulation~\cite{kanatsoulis2018hyperspectralSRR_coupledCPD,prevost2020coupledTucker_hyperspectralSRR_TSP}.


The image fusion problem was formulated in~\cite{kanatsoulis2018hyperspectralSRR_coupledCPD} as a coupled tensor approximation problem. Assuming that the high resolution (HR) image admits a low-rank canonical polyadic decomposition (CPD), the problem was solved using an alternating optimization strategy. The recovery of the correct HR image (HRI) was shown to be guaranteed provided that the CPD of the MSI is identifiable, and state of the art performance was achieved.
A recent work extended this approach by assuming the high resolution images to follow a block term decomposition (BTD), which shows a closer connection to the physical mixing model when compared to the CPD~\cite{zhang2019hyperspectralSRR_coupledBTD}.
A simpler approach was later proposed in~\cite{kanatsoulis2018hyperspSRR_lowRankMatrixTensor} by requiring only the computation of one CPD of the MSI and a singular value decomposition (SVD) of the HSI.


A Tucker decomposition-based approach was later considered in~\cite{prevost2019coupledTucker_hyperspectralSRR,prevost2020coupledTucker_hyperspectralSRR_TSP}. Closed form SVD-based algorithms were proposed for the image fusion problem, achieving results comparable to~\cite{kanatsoulis2018hyperspectralSRR_coupledCPD} at a very small computational complexity; exact recovery guarantees were also provided. A coupled Tucker approximation was also considered in~\cite{li2018hs_ms_fusionTensorFactorization} using an alternating optimization approach and employing a sparsity regularization on the elements of the core tensor. \corange{This work was later extended by incorporating the piecewise smoothness of the reconstructed image tensor by using a Total Variation regularization along each of its modes~\cite{xu2020HSMSfusion_1DTV}.}
Another approach considered the CPD of non-local similar patch tensors to explore the non-local redundancy of the image~\cite{xu2019nonlocalTensorDecompImageFusion}.
\corange{A different non-local approach was also proposed in~\cite{dian2019hyperspectralSRR_tensorMultiRank} by using a recent definition of the nuclear norm of order-3 tensors.}

Most existing algorithms, however, share a common limitation: they assume that the HSI and the MSI are acquired under the same conditions.
However, despite the short revisit cycles provided by the increasing number of optical satellites orbiting the Earth (e.g. Sentinel, Orbview, Landsat and Quickbird missions), the number of platforms carrying both HS and MS sensors is still considerably limited~\cite{eckardt2015DESIS_satelliteSpecs,kaufmann2006EnMAP_satelliteSpecs}. This makes combining HS and MS observations acquired on board of different satellites of great interest to obtain HRIs~\cite{hilker2009fusion_MODIS_denseTimeSeries_exp2,emelyanova2013fusion_MODIS_denseTimeSeries_exp}.
Images acquired at different time instants can be impacted by, e.g., illumination, atmospheric or seasonal changes. This may result in significant variations between the HSI and the MSI~\cite{borsoi2020variabilityReview}, negatively impacting traditional image fusion algorithms.

Recently, a method was proposed to combine HSIs and MSIs accounting for seasonal \corange{(inter-image)} spectral variability~\cite{Borsoi_2018_Fusion}. Using a low-rank matrix formulation, the set of spectral basis vectors of the HRIs underlying the HS and the MS observations are allowed to be different from each other, with variations introduced by a set of multiplicative scaling factors~\cite{imbiriba2018glmm}. This algorithm led to significant performance improvements when the HSI and MSI are subject to spatially uniform seasonal or acquisition variations. However, it does not account for spatially localized changes commonly seen in practical scenes~\cite{borsoi2020variabilityReview}. Moreover, the algorithm in~\cite{Borsoi_2018_Fusion} presented high computation times and does not offer any theoretical guarantees.

In this paper, we propose a tensor-based image fusion formulation that accounts for localized spatial and spectral changes between the HSI and MSI. A general observation model is considered, in which the HRI underlying the MSI admits an additive variability term to account for changes between the scenes. Studying the general identifiability of this model, we show that this variability term can only be identified in general up to its smooth structure (which is defined according to the degradation operators). To introduce additional a priori information and mitigate the ambiguity associated with the proposed model, both the HRI and the additive perturbations are assumed to have low multilinear rank (i.e., to admit a Tucker decomposition). Two algorithms are then proposed, one totally algebraic and another based on an optimization procedure. Theoretical guarantees for the exact recovery of the HRI are provided for both. Simulation results show that the proposed optimization-based algorithm yields superior performance at a considerably lower computational cost when compared to~\cite{Borsoi_2018_Fusion}, especially when spatially localized variability is considered.


\section{Tensors -- background} \label{sec:background}
\subsection{Notation and definitions}
\label{sec:back_defs}


An order-3 tensor $\tensor{T}\in\amsmathbb{R}^{N_1\times N_2 \times N_3}$ is an $N_1\times N_2 \times N_3$ array whose elements are indexed by~$[\tensor{T}]_{n_1,n_2,n_3}$. Each dimension of \cdgreen{an order-3 tensor} is called a \emph{mode}. A mode-$k$ \emph{fiber} of tensor $\tensor{T}$ is the one-dimensional subset of $\tensor{T}$ which is obtained by fixing \cdgreen{all but one of the three modes -- the $k$-th dimension}. Similarly, a \emph{slab} or \emph{slice} of a tensor $\tensor{T}$ is a matrix whose elements are the two-dimensional subset of $\tensor{T}$ obtained by fixing all but two of its modes. Operator $\vect(\cdot)$ represents the standard matrix column-major vectorization, or tensor vectorization. The (left) pseudo-inverse of matrix $\bX$ is denoted by $\bX^\dagger$. We denote scalars by lowercase ($x$) or uppercase ($X$) plain font, vectors and matrices by lowercase ($\bx$) and uppercase ($\bX$) bold font, respectively, and order-3 tensors by calligraphic plain font ($\tensor{T}$) or using the blackboard Greek alphabet ($\tensor{\Theta}$). In the following, we review some useful operations of multilinear (tensor) algebra that will be used in the rest of the manuscript (see, e.g.,~\cite{cichocki2015tensor,kolda2009tensor} for more details).


\begin{definition}
\noindent
{\color{red}The \textbf{mode-$k$ product} between a tensor $\tensor{T}$ and a matrix $\cb{B}$ produces a tensor $\tensor{U}$ that is evaluated such that each mode-$k$ fiber of $\tensor{T}$ is multiplied by $\cb{B}$.
For instance, the mode-2 product between $\tensor{T} \in \amsmathbb{R}^{N_1\times N_2\times N_3}$ and $\cb{B} \in \amsmathbb{R}^{M_2\times N_2}$ produces a tensor 
$\tensor{U}\in\amsmathbb{R}^{N_{1}\times M_2\times N_{3}}$, denoted by 
$\tensor{U}=\tensor{T}\times_2\cb{B}$.}
{\color{magenta} 
Its elements are accessed as $[\tensor{U}]_{n_1,m_2,n_3}=\sum_{i=1}^{N_2}[\tensor{T}]_{n_1,i,n_3} [\cb{B}]_{m_2,i}\ , \ m_2=1,\ldots,M_2$.}
\end{definition}
{\color{magenta}Note that the mode-$k$ product has the following properties: 
\begin{align}
    \tensor{T} \times_i\bA \times_j\bB &= \tensor{T} \times_j\bB  \times_i \bA\,, \ i \neq j \,,
    \label{eq:mode_k_prop1}
    \\
    \tensor{T} \times_k \bA \times_k\bB &= \tensor{T}\times_k\big(\bA\bB \big) \,.
    \label{eq:mode_k_prop2}
\end{align}}
\begin{definition}


\cdgreen{The full \textbf{multilinear product} is denoted by $\big\ldbrack\tensor{T};\bB_{1},\bB_{2},\bB_{3}\big\rdbrack$, and consists of a series of successive mode-$k$ products, for \corange{$k\in\{1,2,3\}$, between a tensor $\tensor{T}$ and matrices $\bB_1$, $\bB_2$ and $\bB_3$,} respectively, and \corange{is} expressed as $\tensor{T}\times_1\bB_{1}\times_2\bB_{2}\times_3\bB_{3}$.}
\end{definition}
%
%
%
\begin{definition}
The \textbf{mode-$k$ matricization} of an order-3 tensor $\tensor{T}\in\amsmathbb{R}^{N_1\times N_2\times N_3}$, denoted by $\unfold{\bT}{k}$, arranges its mode-$k$ fibers to be the columns of the resulting matrix $\unfold{\bT}{k}\in\amsmathbb{R}^{N_k\times N_{\ell}N_{m}}$, $k,\ell,m\in\{1,2,3\}$, $k \neq \ell \neq m$, where the $n_k$-th row of $\unfold{\bT}{k}$ consists of the vectorization of the slice of $\tensor{T}$ obtained by fixing the index of the $k$-th mode of $\tensor{T}$ as $n_k$.
\end{definition}

\begin{definition}
We define by $\tSVD_R(\bX)$ the operator which returns a matrix containing {\color{red} the $R$ left singular vectors associated with the largest singular values of the matrix $\bX$}.
\end{definition}



\subsection{Tensor decompositions}

\corange{The Tucker decomposition is able to represent an order-3 tensor $\tensor{T}\in\amsmathbb{R}^{N_1\times N_2\times N_3}$ compactly, using a set of factor matrices given by $\bB_{i}\in\amsmathbb{R}^{N_i\times K_i}$, $i\in\{1,2,3\}$ and a small core tensor $\tensor{G}\in\amsmathbb{R}^{K_1\times K_2\times K_3}$, as~\cite{kolda2009tensor}
\begin{align}
    \tensor{T} &= \big\ldbrack \tensor{G};\bB_{1},\bB_{2},\bB_{3} \big\rdbrack \,.
    \label{eq:tensorDef_Tucker}
\end{align}
The tuple $(K_1,K_2,K_3)$ is called the multilinear rank of $\tensor{T}$. Each value $K_i$ is also equal to the rank of the mode-$i$ unfolding of~$\tensor{T}$~\cite{cichocki2015tensor}.}

%
%

{\color{magenta} 
The Tucker decomposition 
allows the rank along each mode of the tensor to be different (i.e., $K_i\neq K_j$)~\cite{sidiropoulos2017tensor}. This property can be very useful since it allows one to set a higher rank to specific modes of the decomposition in order to adequately represent the data diversity while still keeping the \mbox{model low rank.}
}

The matricizations and vectorization of a tensor $\tensor{T}$ following the Tucker decomposition~\eqref{eq:tensorDef_Tucker} are given by~\cite{kolda2009tensor}:
\begin{align}
    \vect(\tensor{T}) &= (\bB_{3} \otimes \bB_{2} \otimes \bB_{1} ) \vect(\cmag{\tensor{G}}) \,,
    \label{eq:gen_tensor_vec}
    \\
    \unfold{\bT}{1} &= \bB_{1} \unfold{\cmag{\bG}}{1} (\bB_{3} \otimes \bB_{2})^\top \,,
    \label{eq:gen_tensor_mat_1}
    \\
    \unfold{\bT}{2} &= \bB_{2} \unfold{\cmag{\bG}}{2} (\bB_{3} \otimes \bB_{1})^\top \,,
    \label{eq:gen_tensor_mat_2}
    \\
    \unfold{\bT}{3} &= \bB_{3} \unfold{\cmag{\bG}}{3} (\bB_{2} \otimes \bB_{1})^\top \,.
    \label{eq:gen_tensor_mat_3}
\end{align}
%
%
\corange{The Tucker decomposition can be computed using fast algorithms such as the high-order SVD~\cite{lathauwer2000multilinearHOSVD}.}


\corange{Another classic tensor decomposition is the Canonical Polyadic Decomposition (CPD)~\cite{kolda2009tensor}. The CPD enjoys uniqueness properties under very mild conditions, and its rank can exceed the dimensions of the tensor. However, it is also more difficult to compute, and the same rank value is used to represent all modes of the tensor.}

{\color{magenta}
The Block Term Decomposition (BTD) generalizes CPD and Tucker, and allows us to combine benefits from both approaches~\cite{lathauwer2008tensor_BTD2_uniqueness}.}
Specifically, the BTD of an order-3 tensor $\tensor{T}$ is defined as a sum of rank-$(K_{1,r},K_{2,r},K_{3,r})$ terms as~\cite{lathauwer2008tensor_BTD2_uniqueness}:
\begin{align} \label{eq:tensor_BTD_i}
    \tensor{T} & = \sum_{r=1}^R \big\ldbrack\tensor{G}_r;\bB_{r,1},\bB_{r,2},\bB_{r,3} \big\rdbrack \,,
\end{align}
where $\tensor{G}_r$, $\bB_{r,1}$, $\bB_{r,2}$, and $\bB_{r,3}$, for $r=1,\ldots,R$, are the core tensors and the factors corresponding to each mode of $\tensor{T}$. \corange{Differently from the Tucker decomposition, the BTD additionally requires the selection of parameter $R$ (number of blocks), and can be more costly to compute.}
{\color{magenta}
However, the BTD benefits from uniqueness results which, although not as strong than those of the CPD, are still interesting for many applications~\cite[Section~5]{lathauwer2008tensor_BTD2_uniqueness}.} 
\corange{These results will prove very important to derive the recoverability guarantees for the algorithm in Section~\ref{sec:alg2_optimizationBased}.
Note that in the CPD can be viewed as a special case of the BTD with $R=1$.}


\begin{figure}
    \centering
    \includegraphics[width=\linewidth]{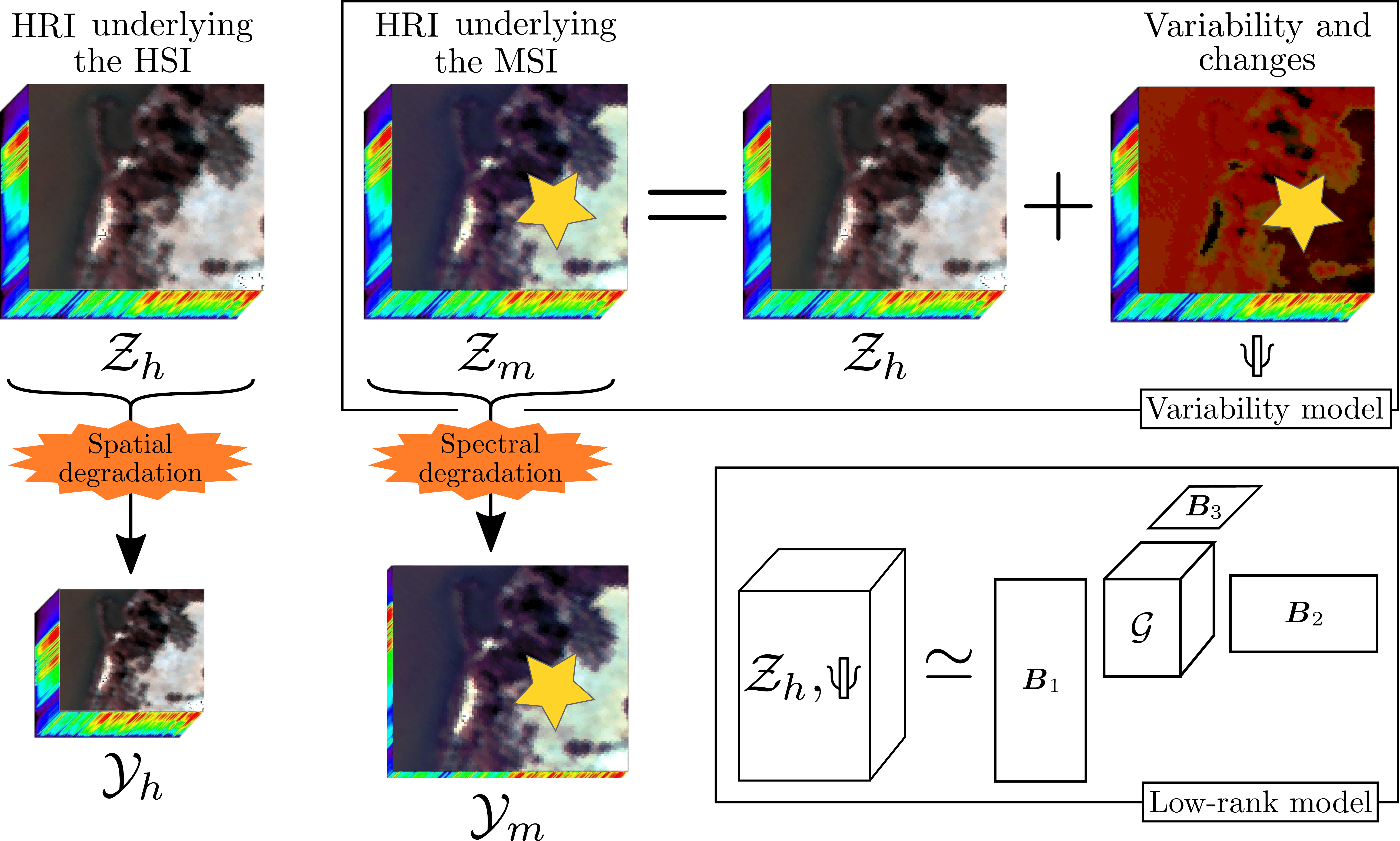}
    \vspace{-0.7cm}
    \caption{\corange{Proposed imaging model. The observed HSI and MSI are acquired under different conditions or at different time instants. Their underlying HRIs, represented by tensors $\tensor{Z}_h$ and $\tensor{Z}_m$, respectively, can be different from each other due to the effect of spectral and acquisition variations, as well as scenery changes. The changes occurring between the HRIs $\tensor{Z}_h$ and $\tensor{Z}_m$ are represented using an additive tensor $\tensor{\Psi}$ which captures variability and changes, and both the HRI and the variability tensors are assumed to have low Tucker rank.}}
    \label{fig:gen_diagram}
\end{figure}

\section{Proposed model and its undeterminacies}
\label{sec:imaging_mdl_identif}

\subsection{The imaging model}

Let an HSI with high spectral resolution and low spatial resolution be represented as an order-3 tensor $\tensor{Y}_h\in\amsmathbb{R}^{N_1\times N_2\times L_h}$, where $N_1$ and $N_2$ are the spatial and $L_h$ the spectral dimensions. 
Similarly, an MSI with high spatial and low spectral resolution is denoted by an order-3 tensor $\tensor{Y}_m\in\amsmathbb{R}^{M_1\times M_2\times L_m}$, where $M_1>N_1$ and $M_2>N_2$ are the spatial and $L_m<L_h$ the spectral dimensions.
%
%
Both the HSI and the MSI are assumed to be degraded versions of a tensor $\tensor{Z}\in\amsmathbb{R}_+^{M_1\times M_2\times L_h}$, with high spectral and spatial resolutions.
This degradation process is commonly described as~\cite{kanatsoulis2018hyperspectralSRR_coupledCPD,prevost2019coupledTucker_hyperspectralSRR,prevost2020coupledTucker_hyperspectralSRR_TSP,li2018hs_ms_fusionTensorFactorization}:
\begin{align}
    \tensor{Y}_h &= \tensor{Z} \times_1 \bP_1 \times_2 \bP_2 + \tensor{E}_h \,,
    \label{eq:tensor_obs_mdl_hsi_novariability}
    \\
    \tensor{Y}_m &= \tensor{Z} \times_3 \bP_3 + \tensor{E}_m  \,,
    \label{eq:tensor_obs_mdl_msi_novariability}
\end{align}
where tensors $\tensor{E}_m\in\amsmathbb{R}^{M_1\times M_2\times L_m}$ and $\tensor{E}_h\in\amsmathbb{R}^{N_1\times N_2\times L_h}$ represent additive noise.
The matrix $\bP_3\in\amsmathbb{R}^{L_m\times L_h}$ contains the spectral response functions (SRF) of each band of the multispectral sensor, and matrices $\bP_1\in\amsmathbb{R}^{N_1\times M_1}$ and $\bP_2\in\amsmathbb{R}^{N_2\times M_2}$ represent the spatial blurring and downsampling in the hyperspectral sensor, which we assume to be separable for each spatial dimension as previously done in, e.g.,~\cite{kanatsoulis2018hyperspectralSRR_coupledCPD,prevost2019coupledTucker_hyperspectralSRR,prevost2020coupledTucker_hyperspectralSRR_TSP,li2018hs_ms_fusionTensorFactorization}. 
\corange{Note that since the mode-$k$ product obeys~\eqref{eq:mode_k_prop1}, the choice of the ordering of $\bP_1$ and $\bP_2$ in~\eqref{eq:tensor_obs_mdl_hsi_novariability} does not affect the result.}
To make notation more convenient, we also denote the (linear) spatial and spectral degradation operators more compactly as
\begin{align}
    \opPspat(\tensor{T}) &= \tensor{T} \times_1 \bP_1 \times_2 \bP_2 \,,
    \label{eq:operatorNotation_a}
    \\
    \opPspec(\tensor{T}) &= \tensor{T} \times_3 \bP_3 \,.
    \label{eq:operatorNotation_b}
\end{align}



Most previous works consider that $\tensor{Y}_h$ and $\tensor{Y}_m$ are acquired under the same conditions, implicitly assuming that no variability \corange{or changes} occur between the images.
\corange{However, when the HSI and MSI are not acquired from the same mission/instrument and at the same time, the scene which underlies the (degraded) observations $\tensor{Y}_h$ and $\tensor{Y}_m$ can be subject to significant changes, referred to as {\color{magenta} inter-image} variability\footnote{\color{magenta} Here, variability should not be confused with the spectral variability considered in \cite{borsoi2020variabilityReview}, which focuses on inter-pixel endmembers variations.}, which include spatial and spectral variations as illustrated in Fig.~\ref{fig:gen_diagram}. Spectral variations originate from different \cmag{chemical}, atmospheric, illumination or seasonal conditions between the scenes~\cite{somers2011variabilityReview,Zare-2014-ID324-variabilityReview,borsoi2020variabilityReview}, and are typical even under short acquisition time differences. Spatial variations, on the other hand, occur due to, e.g., some regions of the scene being affected unequally by seasonal effects (which are strongly material-dependent~\cite{borsoi2020variabilityReview}) or due to the sudden insertion/removal of an object~\cite{liu2019reviewCD_GRSM}. Spatial variations  can be very prominent when large acquisition time differences are considered.
These effects are not accounted for in the majority of the existing algorithms, what motivates the development of more flexible models.
}

Recently, spatially uniform spectral variability has been considered in~\cite{Borsoi_2018_Fusion}. The image fusion problem was formulated as a matrix factorization problem, and the (multiplicative) spectral variability as well as the spatial coefficients were estimated from the observed images.
However, this work still did not address two fundamental problems: 1) How to account for both spatial and spectral variability and 2) what theoretical guarantees can be offered for the recovery of the HRI and (possibly) of the variability factors under these more challenging conditions.

To address these issues, we adopt a more general approach by considering two different HRIs $\tensor{Z}_h\in\amsmathbb{R}_+^{M_1\times M_2\times L_h}$ and $\tensor{Z}_m\in\amsmathbb{R}_+^{M_1\times M_2\times L_h}$, both with high spectral and spatial resolutions, underlying the observed HSI and the MSI, respectively. This leads to the following extension of model~\eqref{eq:tensor_obs_mdl_hsi_novariability}--\eqref{eq:tensor_obs_mdl_msi_novariability}:
\begin{align}
    \tensor{Y}_h &= \opPspat(\tensor{Z}_h) + \tensor{E}_h \,,
    \label{eq:tensor_obs_mdl_hsi_0}
    \\
    \tensor{Y}_m &= \opPspec(\tensor{Z}_m) + \tensor{E}_m  \,.
    \label{eq:tensor_obs_mdl_msi_0}
\end{align}


{\color{magenta} The tensors $\tensor{Z}_h$ and $\tensor{Z}_m$ \cdgreen{represent} the underlying HRIs of the observed scene under different acquisition conditions and at possibly different times. To account for inter-image variability, the HRIs are related to each other as follows:}
\begin{align}
    \tensor{Z}_m = \corange{\tensor{Z}_h} + \tensor{\Psi},
    \label{eq:hr_hsi_msi_add_mdl_0}
\end{align}
where $\tensor{\Psi}\in\amsmathbb{R}^{M_1\times M_2\times L_h}$ is an additive variability tensor representing changes between the scenes. 
%
{\color{magenta}By introducing $\tensor{\Psi}$, model~\eqref{eq:hr_hsi_msi_add_mdl_0} makes the inter-image variability between the HSI and the MSI explicit.}

Considering the variability model~\eqref{eq:hr_hsi_msi_add_mdl_0} along with~\eqref{eq:tensor_obs_mdl_hsi_0}--\eqref{eq:tensor_obs_mdl_msi_0}, we obtain the following observation model for the acquired HSI and MSI:
\begin{align}
    \tensor{Y}_h &= \opPspat(\corange{\tensor{Z}_h}) + \tensor{E}_h \,,
    \label{eq:tensor_obs_mdl_hsi}
    \\
    \tensor{Y}_m &= \opPspec(\corange{\tensor{Z}_h} + \tensor{\Psi}) + \tensor{E}_m.
    \label{eq:tensor_obs_mdl_msi_add}
\end{align}
{\color{magenta}In the following, $\tensor{Z}_h$ and $\tensor{\Psi}$ will be referred to as the HRI and the variability tensor, respectively.}

\subsection{The image fusion problem and its undeterminacies}

The image fusion problem in this case consists in recovering~$\tensor{\Psi}$ and $\corange{\tensor{Z}_h}$ from the observed images $\tensor{Y}_h$ and $\tensor{Y}_m$. More precisely,
\begin{align} \label{eq:HS_MS_fus_var_prob}
\Bigg\{
\begin{split}
    & \text{find } \corange{\tensor{Z}_h}\in\Omega_Z \text{ and }  \tensor{\Psi}\in\Omega_{\Psi},
    \\
    & \text{such that equations~\eqref{eq:tensor_obs_mdl_hsi}--\eqref{eq:tensor_obs_mdl_msi_add} are satisfied.}
\end{split}
\end{align}
The sets $\Omega_Z\subseteq\amsmathbb{R}^{M_1\times M_2\times L_h}$  and $\Omega_{\Psi}\subseteq\amsmathbb{R}^{M_1\times M_2\times L_h}$ denote prior information about the HRI and the variability factor, respectively.

Since the number of unknowns is significantly greater than the number of observations, problem~\eqref{eq:HS_MS_fus_var_prob} is severely ill-posed and additional \emph{a priori} information about the structure of $\corange{\tensor{Z}_h}$ and $\tensor{\Psi}$ must be introduced through the sets $\Omega_Z$ and $\Omega_{\Psi}$ in order to obtain a stable recovery. 
Common information that has been used to construct $\Omega_Z$ includes spatial \mbox{(piecewise-)} smoothness~\cite{simoes2015HySure}, low matrix (spectral) rank~\cite{yokoya2012coupledNMF,Borsoi_2018_Fusion}, low tensor rank~\cite{kanatsoulis2018hyperspectralSRR_coupledCPD,prevost2020coupledTucker_hyperspectralSRR_TSP,xu2019nonlocalTensorDecompImageFusion}, and non-local spatial information~\cite{li2018hs_ms_fusionTensorFactorization}.

The choice of prior information in $\Omega_Z$ and $\Omega_{\Psi}$ turns to the question of whether assuming additional structure over the pair $(\corange{\tensor{Z}_h},\tensor{\Psi})$ makes these variables \emph{identifiable} from the observations $(\tensor{Y}_{h},\tensor{Y}_{m})$.
Recent works in HS-MS image fusion advocates for a low-rank tensor model~\cite{kanatsoulis2018hyperspectralSRR_coupledCPD,prevost2020coupledTucker_hyperspectralSRR_TSP,prevost2019coupledTucker_hyperspectralSRR}. However, the case at hand is more challenging because of the additional variability $\tensor{\Psi}$, which makes the model more ambiguous.


In many inverse problems such as matrix or tensor factorization, dictionary learning and blind deconvolution, identifiability of the underlying variables often can only be defined up to some fundamental ambiguities.
Transformation groups and equivalence classes~\cite{li2016identifiabilityBilinearInverseProbs} can be used to precisely define which sets of solutions can generate each possible observations.
These ideas can be leveraged to characterize some of the fundamental ambiguities associated with the model~\eqref{eq:tensor_obs_mdl_hsi}--\eqref{eq:tensor_obs_mdl_msi_add}, and to provide insights into the development of efficient algorithms.
First, we will show that the presence of $\tensor{\Psi}$ makes the model fundamentally ambiguous, as the content in $\corange{\tensor{Z}_h}$ cannot be easily distinguished from that of $\tensor{\Psi}$. Moreover, we will define an equivalence class that characterizes the sets of images $\corange{\tensor{Z}_h}$ and factors $\tensor{\Psi}$ which are certain to result in different observed HSI and MSI.
This gives us insight into what kind of structure from these variables can be recovered from $(\tensor{Y}_{h},\tensor{Y}_{m})$.

\corange{To proceed, let us first denote by $\mathscr{A}:(\tensor{Z}_h,\tensor{\Psi})\mapsto(\tensor{Y}_h,\tensor{Y}_m)$ the operator which describes the degradation process in equations~\eqref{eq:tensor_obs_mdl_hsi}--\eqref{eq:tensor_obs_mdl_msi_add}.} 
\corange{By representing operators $\opPspat$, $\opPspec$ and $\mathscr{A}$ in matrix form as $\widetilde{\bP}_{1,2}\in\amsmathbb{R}^{N_1N_2L_h \times M_1M_2L_h}$, $\widetilde{\bP}_{3}\in\amsmathbb{R}^{M_1M_2L_m \times M_1M_2L_h}$ and $\bA\in\amsmathbb{R}^{(N_1N_2L_h+M_1M_2L_m)\times 2M_1M_2L_h}$, respectively, we can write the model~\eqref{eq:tensor_obs_mdl_hsi}--\eqref{eq:tensor_obs_mdl_msi_add} in the noiseless case ($\tensor{E}_h = \tensor{E}_m = \tensor{0}$) equivalently as: }
\begin{align}
    \begin{bmatrix} 
    \vect(\tensor{Y}_h) \\ \vect(\tensor{Y}_m)
    \end{bmatrix}
    = \underbrace{
    \begin{bmatrix} 
    \widetilde{\bP}_{1,2} & \cb{0} \\
    \widetilde{\bP}_{3} & \widetilde{\bP}_{3}
    \end{bmatrix}
    }_{\bA}
    \begin{bmatrix} 
    \vect(\tensor{Z}_h) \\ \vect(\tensor{\Psi})
    \end{bmatrix} \,.
    \label{eq:thm1_mtxForm_1}
\end{align}

\corange{Define also the equivalence relation $\simZPsi$ based on operator $\mathscr{A}$ as follows:
\begin{align}
    (\tensor{Z}_h,\tensor{\Psi}) \simZPsi (\tensor{Z}_h',\tensor{\Psi}')  \iff  \mathscr{A}(\tensor{Z}_h,\tensor{\Psi}) - \mathscr{A}(\tensor{Z}_h',\tensor{\Psi}') =\tensor{0}
    \label{eq:thm_prel2_eq_rel}
\end{align}
and its associated equivalence class (EC) $[(\tensor{Z}_h,\tensor{\Psi})]_{\simZPsi}$ as
\begin{align}
    [(\tensor{Z}_h,\tensor{\Psi})]_{\simZPsi} & = \big\{\tensor{X}\in\Omega_Z\times\Omega_{\Psi}\,:\,\tensor{X} \simZPsi (\tensor{Z}_h,\tensor{\Psi})\big\}.
    \label{eq:thm_prel2_eq_class}
\end{align}
Now, we are ready to present the following result.}

\begin{theorem} \label{thm:identifiability_ECc_psi}
Suppose that the observation noise is zero (i.e., $\tensor{E}_h=\tensor{E}_m=\tensor{0}$) and that $\Omega_Z=\amsmathbb{R}^{M_1\times M_2\times L_h}_+$ and $\Omega_{\Psi}=\amsmathbb{R}^{M_1\times M_2\times L_h}$. 
Then, given a set of HSI and MSI observations $(\tensor{Y}_h,\tensor{Y}_m)$\corange{, the following is verified:}
\begin{itemize}
    \item[a)] \corange{If operator $\mathscr{A}$ has nontrivial nullspace (e.g., if $L_m<L_h$ or if $N_1N_2<M_1M_2$), then the pair $(\tensor{Z}_h,\tensor{\Psi})$ cannot be uniquely identified from the observations $(\tensor{Y}_h,\tensor{Y}_m)$.}

    \item[b)] There is only one (unique) equivalence class $[(\tensor{Z}_0,\tensor{\Psi}_0)]_{\simZPsi}$ containing HR images and scaling factors $(\corange{\tensor{Z}_h},\tensor{\Psi})$ that can generate~$(\tensor{Y}_h,\tensor{Y}_m)$ according to model~\eqref{eq:tensor_obs_mdl_hsi}--\eqref{eq:tensor_obs_mdl_msi_add}. In other words, $(\corange{\tensor{Z}_h},\tensor{\Psi})$ can be identified uniquely up to $[(\tensor{Z}_0,\tensor{\Psi}_0)]_{\simZPsi}$.
\end{itemize}

\end{theorem}

\begin{proof}
\corange{\textbf{Proof of a):}} \corange{Due to the special structure of matrix $\bA$ in~\eqref{eq:thm1_mtxForm_1}, is is clear that $\rank(\bA)=\rank(\widetilde{\bP}_{1,2})+\rank(\widetilde{\bP}_{3})\leq 2M_1M_2L_h$. Thus, if $\mathscr{A}$ has nontrivial nullspace then either $\opPspat$ or $\opPspec$ have nontrivial nullspace.}
If the operator $\opPspat$ has nontrivial nullspace, then we can find $\corange{\tensor{Z}_h}$, $\corange{\tensor{Z}_h'}$, different from one another, such that
\begin{align}
    \opPspat(\corange{\tensor{Z}_h}) & = \opPspat(\corange{\tensor{Z}_h'}) \,,
    \label{eq:thm1_i}
\end{align}
implying that $\tensor{Y}_h=\tensor{Y}_h'$. Now, we can always find $\tensor{\Psi}$ and $\tensor{\Psi}'$ satisfying
\begin{align}
    \corange{\tensor{Z}_h} + \tensor{\Psi} & = \corange{\tensor{Z}_h'} + \tensor{\Psi}',
    \label{eq:thm1_ii}
\end{align}
which implies that $\tensor{Y}_m=\tensor{Y}_m'$ (i.e., the model is not identifiable). 

Similarly, if operator $\opPspec$ has nontrivial nullspace, then suppose we select $\corange{\tensor{Z}_h}=\corange{\tensor{Z}_h'}$. This makes $\tensor{Y}_h=\tensor{Y}_h'$. Then, we can select $\tensor{\Psi}$ and $\tensor{\Psi}'$, distinct from one another, satisfying
\begin{align}
    \opPspec(\tensor{\Psi}) - \opPspec(\tensor{\Psi}') = \tensor{0} \,,
    \label{eq:thm1_iii}
\end{align}
where $\tensor{0}$ is the tensor of zeros. Since $\corange{\tensor{Z}_h}=\corange{\tensor{Z}_h'}$, this leads to:
\begin{align}
    \opPspec(\corange{\tensor{Z}_h}) & = \opPspec(\corange{\tensor{Z}_h'})
    \nonumber\\
    & = \opPspec(\corange{\tensor{Z}_h'}) + \opPspec(\tensor{\Psi}) - \opPspec(\tensor{\Psi}') \,,
    \label{eq:thm1_iv}
\end{align}
which also implies that $\tensor{Y}_m=\tensor{Y}_m'$ (i.e., the model is not identifiable).

\corange{\textbf{Proof of b):}} Note that since there is no additive noise, $\mathscr{A}$ is a linear operator and, thus, satisfies the following relation:
\begin{align}
    \mathscr{A}(\corange{\tensor{Z}_h},\tensor{\Psi}) - \mathscr{A}(\corange{\tensor{Z}_h'},\tensor{\Psi}') =
    \mathscr{A}(\corange{\tensor{Z}_h}-\corange{\tensor{Z}_h'},\tensor{\Psi}-\tensor{\Psi}') \,.
\end{align}
By inspecting the definition of the equivalence relation in~\eqref{eq:thm_prel2_eq_rel}, it can be seen that the equivalence class in~\eqref{eq:thm_prel2_eq_class} is characterized by the
kernel of the operator $\mathscr{A}$, and can also be written as:
\begin{align}
    [(\corange{\tensor{Z}_h},\tensor{\Psi})]_{\simZPsi} 
    & = \big\{(\corange{\tensor{Z}_h},\tensor{\Psi}) + \tensor{X} \,:\, \tensor{X} \in \ker(\mathscr{A})\big\}.
    \label{eq:thm_prel2_eq_class2}
\end{align}

Now, suppose that we have two sets of HR images and variability factors belonging to different ECs, i.e., $(\corange{\tensor{Z}_h},\tensor{\Psi})\in[(\tensor{Z}_0,\tensor{\Psi}_0)]_{\simZPsi}$, $(\corange{\tensor{Z}_h'},\tensor{\Psi}')\in[(\tensor{Z}_0',\tensor{\Psi}_0')]_{\simZPsi}$, with $[(\tensor{Z}_0,\tensor{\Psi}_0)]_{\simZPsi}\neq[(\tensor{Z}_0',\tensor{\Psi}_0')]_{\simZPsi}$. Comparing the observations $(\tensor{Y}_h,\tensor{Y}_m)$ and $(\tensor{Y}_h',\tensor{Y}_m')$ generated by elements of each EC, we have:
\begin{align}
    (\tensor{Y}_h,\tensor{Y}_m) 
    &= \mathscr{A}((\corange{\tensor{Z}_h},\tensor{\Psi}) + \tensor{X})
    \nonumber \\
    & \neq \mathscr{A}((\corange{\tensor{Z}_h'},\tensor{\Psi}')  + \tensor{X}')
    =  (\tensor{Y}_h',\tensor{Y}_m') \,,
\end{align}
for all $\tensor{X},\tensor{X}'\in\ker(\mathscr{A})$. Thus, elements selected from different ECs will always lead to different observations, ensuring that the EC is identifiable.
\end{proof}


Intuitively, \corange{item a)} shows that variations in $\tensor{\Psi}$ that occur in the nullspace of operators $\opPspat$ and $\opPspec$ are not reflected in the corresponding observations $(\tensor{Y}_h,\tensor{Y}_m)$ in such a way that they can be differentiated from possible changes in $\corange{\tensor{Z}_h}$.
More generally, changes of $\tensor{\Psi}$ and $\corange{\tensor{Z}_h}$ that occur in the nullspace of matrix $\bA$ do not affect $(\tensor{Y}_h,\tensor{Y}_m)$\corange{, which is clear from~\eqref{eq:thm1_mtxForm_1}}. This notion can be further extended by noting that changes occurring in the column space of $\bA$ will certainly lead to different observations, which is made precise in \corange{item b)}.

\corange{Item b) in} Theorem~\ref{thm:identifiability_ECc_psi} allows us to characterize the ambiguities in the model \corange{in more detail}. However, it is important to consider the characteristics of $\ker(\mathscr{A})$ in our problem to better understand the recoverability of the variability factor $\tensor{\Psi}$. Let us consider the model in~\eqref{eq:thm1_mtxForm_1}, and two sets of variables $(\corange{\tensor{Z}_h},\tensor{\Psi})\simZPsi(\corange{\tensor{Z}_h'},\tensor{\Psi}')$, belonging to the same EC.
It can be seen that to generate the same observations, the HR images need to satisfy $\corange{\tensor{Z}_h}-\corange{\tensor{Z}_h'}\in\ker(\opPspat)$, while the variability factors have to satisfy $\opPspec(\corange{\tensor{Z}_h}-\corange{\tensor{Z}_h'})=-\opPspec(\tensor{\Psi}-\tensor{\Psi}')$. Therefore, the general form of the difference between the variability factors inside each equivalence class is of the form:
\begin{align}
    \tensor{\Psi}-\tensor{\Psi}' = -\underbrace{(\corange{\tensor{Z}_h}-\corange{\tensor{Z}_h'})}_{\in\ker(\opPspat)} + \underbrace{\tensor{X}}_{\in\ker(\opPspec)} \,.
    \label{eq:psi_EC_general_form}
\end{align}
The set of all possible $\tensor{\Psi}-\tensor{\Psi}'$ satisfying~\eqref{eq:psi_EC_general_form} is the sum of $\ker(\opPspec)$ and $\ker(\opPspat)$, which is given by $\ker(\opPspec\circ\opPspat)$. We can readily see that $\tensor{\Psi}$ cannot be recovered from the observations.
Only the spectrally degraded variability factors $\opPspec(\tensor{\Psi})$ can be uniquely recovered (which comes ``for free'' with the recovery of $\corange{\tensor{Z}_h}$ since it can be computed as $\opPspec(\tensor{\Psi})=\tensor{Y}_m-\opPspec(\corange{\tensor{Z}_h})$). This makes it sufficient to study the capability of an algorithm to recover $\corange{\tensor{Z}_h}$ in our model. Since the matrices $\bP_i$, $i\in\{1,2,3\}$ are essentially low-pass filtering and downsampling operators, their nullspaces intuitively encode high-frequency information along each tensor mode. Thus, only the smooth structure of $\tensor{\Psi}$ can be identified uniquely from observations~$(\tensor{Y}_h,\tensor{Y}_m)$, since otherwise we cannot separate the effects of $\tensor{\Psi}$ from $\corange{\tensor{Z}_h}$.

We also note that each EC in~\eqref{eq:thm_prel2_eq_class}, which contains all factors $\tensor{\Psi}$ whose difference lies in the nullspace of the combined operator $\opPspat\circ\opPspec$, is strictly larger than if we considered changes that occur in the nullspace of each of these operators individually (i.e., $\opPspat$ and $\opPspec$). 

Theorem~\ref{thm:identifiability_ECc_psi} guarantees that tensors belonging to different ECs will result in different observations, which is the minimal requirement for having identifiable $\corange{\tensor{Z}_h}$ and $\opPspec(\tensor{\Psi})$. However, the coresponding inverse problem still remains ill-posed as the number of {\color{red}unknowns} is greater than the number of observations.
Thus, stronger identifiability conditions cannot be obtained unless we provide stricter a priori characterizations of the sets $\Omega_Z$ and $\Omega_{\Psi}$.

\subsection{A Low-Multilinear-Rank Model}
\label{sec:low_rank_model}

One possible condition that can be imposed on the structures of both $\Omega_{\Psi}$ and $\Omega_Z$ is the low-rank tensor model. This kind of structure makes it possible to obtain identifiability and exact recovery guarantees for problem~\eqref{eq:HS_MS_fus_var_prob}, where spatial and spectral variabilities are present. Moreover, it also makes the problem well-posed and easier to solve since the number of unknowns becomes smaller than the amount of available data.

Suppose that $\corange{\tensor{Z}_h}$ and $\tensor{\Psi}$ have multilinear ranks $(K_{Z,1},K_{Z,2},K_{Z,3})$ and $(K_{\Psi,1},K_{\Psi,2},K_{\Psi,3})$, respectively. This means that they can be represented as 
\begin{align}
    \corange{\tensor{Z}_h} &= \big\ldbrack\tensor{G}_{Z}; \bB_{Z,1}, \bB_{Z,2}, \bB_{Z,3} \big\rdbrack \,,
    \label{eq:tucker_mdl_Z}
    \\
    \tensor{\Psi} &= \big\ldbrack\tensor{G}_{\Psi}; \bB_{\Psi,1}, \bB_{\Psi,2}, \bB_{\Psi,3} \big\rdbrack \,,
    \label{eq:tucker_mdl_Psi}
\end{align}
where $\bB_{Z,i}\in\amsmathbb{R}^{M_i\times K_{Z,i}}$, $\bB_{\Psi,i}\in\amsmathbb{R}^{M_i\times K_{\Psi,i}}$, $i\in\{1,2\}$, $\bB_{Z,3}\in\amsmathbb{R}^{L_h\times K_{Z,3}}$, $\bB_{\Psi,3}\in\amsmathbb{R}^{L_h\times K_{\Psi,3}}$ are the factor matrices and $\tensor{G}_{Z}\in\amsmathbb{R}^{K_{Z,1}\times K_{Z,2}\times K_{Z,3}}$, $\tensor{G}_{\Psi}\in\amsmathbb{R}^{K_{\Psi,1}\times K_{\Psi,2}\times K_{\Psi,3}}$ are the core tensors.

Our objective is to study the identifiability and exact recovery of these variables given the observation model in~\eqref{eq:tensor_obs_mdl_hsi}--\eqref{eq:tensor_obs_mdl_msi_add}.
Using this model\corange{, and applying the definition of the multilinear product and the properties of the mode-$k$ product defined in Section~\ref{sec:back_defs}}, the noiseless case of the degradation model~\eqref{eq:tensor_obs_mdl_hsi}--\eqref{eq:tensor_obs_mdl_msi_add} can be written as
\begin{align}
    \tensor{Y}_h ={}&  \big\ldbrack\tensor{G}_{Z}; \bP_1 \bB_{Z,1}, \bP_2 \bB_{Z,2}, \bB_{Z,3} \big\rdbrack \,,
    \label{eq:obs_mdl_hsi_lowrank}
    \\
    \tensor{Y}_m ={}& \big\ldbrack\tensor{G}_{Z}; \bB_{Z,1}, \bB_{Z,2}, \bP_3 \bB_{Z,3} \big\rdbrack 
    \nonumber \\
    & + \big\ldbrack\tensor{G}_{\Psi}; \bB_{\Psi,1}, \bB_{\Psi,2}, \bP_3 \bB_{\Psi,3} \big\rdbrack \,. 
    \label{eq:obs_mdl_msi_lowrank_btd}
\end{align}

Note that we can represent the multispectral image model in~\eqref{eq:obs_mdl_msi_lowrank_btd} equivalently using a standard Tucker model as:
\begin{align}
    \tensor{Y}_m ={}& \big\ldbrack \tensor{C}_{m}; \bC_{m,1}, \bC_{m,2}, \bP_3 \bC_{m,3} \big\rdbrack \,, 
    \label{eq:obs_mdl_msi_lowrank_hightucker}
\end{align}
where $\bC_{m,i}$, $i\in\{1,2,3\}$ and $\tensor{C}_m$ are the factor matrices and the core tensor of the MSI, which satisfy:
\begin{align}
    \tensor{C}_{m} &= \tensor{G}_{Z}\oplus\tensor{G}_{\Psi} \,,   \label{eq:tucker_mdl_msi_var_ext_core}
    \\
    \bC_{m,i} &= \big[\bB_{Z,i} \, \bB_{\Psi,i}\big],\,\,\, i\in\{1,2,3\} \,, \label{eq:tucker_mdl_msi_var_ext_factors}
\end{align}
{\color{red}where for two tensors $\tensor{A}$ and $\tensor{B}$, the binary operation $\tensor{A}\oplus\tensor{B}$ returns a block-diagonal tensor whose diagonal blocks are $\tensor{A}$ and $\tensor{B}$.}

\corange{The model in equations~\eqref{eq:tucker_mdl_Z}--\eqref{eq:tucker_mdl_msi_var_ext_factors} will be subsequently used in Sections~\ref{sec:alg1_algebraic} and~\ref{sec:alg2_optimizationBased} to develop two image fusion algorithms, one algebraic (faster, but with stringent rank constraints) and another based on an optimization procedure (which allows for higher rank values). In each case, a new algorithm will be presented followed by its recoverability guarantees. It should be noted that in practice, the HRI $\tensor{Z}_h$ and the variability tensor $\tensor{\Psi}$ can have high rank. Nevertheless, we will perform only a coupled tensor approximation, with which we are able to capture most of the energy of the data even with insufficient ranks. In practice, however, higher rank models will be preferred to ensure the data is well represented and to avoid the presence artifacts in the reconstructed HRI (which will be achieved with the method of Section~\ref{sec:alg2_optimizationBased}).}

\section{An algebraic algorithm}
\label{sec:alg1_algebraic}

Considering the model in Section~\ref{sec:low_rank_model}, the image fusion problem consists in estimating the factors and core tensor $\tensor{G}_Z$, $\bB_{Z,i}$, $i\in\{1,2,3\}$. However, if the values composing the multilinear rank of $\corange{\tensor{Z}_h}$ are sufficiently low, those variables can be computed by solving the following coupled system of equations:
\begin{align}
\begin{cases}
    \tensor{Y}_h &= \big\ldbrack\tensor{G}_{Z}; \bC_{h,1}, \bC_{h,2}, \bB_{Z,3} \big\rdbrack
    \\
    \tensor{Y}_m &= \big\ldbrack\tensor{C}_{m}; \bC_{m,1}, \bC_{m,2}, \bP_3 \bC_{m,3} \big\rdbrack
    \\
    \bC_{h,i} &= \bP_i \bB_{Z,i}, \,\,\, i\in\{1,2\}
    \\
    \bC_{m,i} &= \big[\bB_{Z,i},\bB_{\Psi,i}\big], \,\,\, i\in\{1,2,3\}
\end{cases},
\label{eq:alg1_coupled_eqs_formulation}
\end{align}
where $\bC_{h,i}$, $i\in\{1,2\}$ denote the spatial factor matrices of the HSI, and the HRI is obtained from the solution of~\eqref{eq:alg1_coupled_eqs_formulation} as $\corange{\tensor{Z}_h}=\ldbrack\tensor{G}_{Z};\bB_{Z,1},\bB_{Z,2},\bB_{Z,3}\rdbrack$. 



If we suppose that $K_{Z,i}+K_{\Psi,i}\leq N_i$, $i\in\{1,2\}$,~\eqref{eq:alg1_coupled_eqs_formulation} can be solved using an efficient, algebraic approach detailed in Alg.~\ref{alg:alg1}, which we call CT-STAR (Coupled Tucker decompositions for hyperspectral Super-resoluTion with vARiability). \corange{The basic intuition behind this algorithm is to use the correspondence between the mode-1 and mode-2 matricizations of the HSI and MSI in order to separate the HRI from the variability tensor when computing its factor matrices. More details will be provided in the following.}

\begin{algorithm} [thb]
\footnotesize
\SetKwInOut{Input}{Input}
\SetKwInOut{Output}{Output}
\caption{\mbox{Algebraic image fusion method (CT-STAR)}\label{alg:alg1}}
\Input{\mbox{Images $\tensor{Y}_h$, $\tensor{Y}_m$ ranks $K_{Z,i}$, $K_{\Psi,i}$, $i\in\{1,2,3\}$}}
\Output{HRI $\corange{\widehat{\tensor{Z}}_h}$, spectrally degraded variability factors $\opPspec(\widehat{\tensor{\Psi}})$}

Check if $K_{Z,i}+K_{\Psi,i}\leq N_i$, $i\in\{1,2\}$ \;

Compute $\widehat{\bC}_{h,3} = \tSVD_{K_{Z,3}}(\unfold{\bY_{\!h}}{3})$ \;

Compute $\widehat{\bC}_{m,i} = \tSVD_{K_{Z,i}+K_{\Psi,i}}(\unfold{\bY_{\!m}}{i})$ for $i\in\{1,2\}$ \;

Compute $\widetilde{\bQ}_i$, for $i\in\{1,2\}$, as $\widetilde{\bQ}_i = \big(\bP_i \widehat{\bC}_{m,i}\big)^{\dagger} \tSVD_{K_{Z,i}}(\unfold{\bY_{\!h}}{i})$\;

Compute $\widetilde{\bC}_{m,i}=\widehat{\bC}_{m,i} \widetilde{\bQ}_i$, for $i\in\{1,2\}$ \;

Compute $\widehat{\tensor{G}}_Z$ by solving
$(\widehat{\bC}_{h,3} \otimes \bP_2\widetilde{\bC}_{m,2} \otimes \bP_1\widetilde{\bC}_{m,1} )\vect(\tensor{G}_Z)=\vect(\tensor{Y}_h)$\;

Compute $\corange{\widehat{\tensor{Z}}_h}=\ldbrack\widehat{\tensor{G}}_{Z}; \widetilde{\bC}_{m,1}, \widetilde{\bC}_{m,2}, \widehat{\bC}_{h,3}\rdbrack$ \;

Compute $\opPspec(\widehat{\tensor{\Psi}})=\tensor{Y}_m-\corange{\widehat{\tensor{Z}}_h}\times_3\bP_3$ \;
\end{algorithm}

It is important to note that CT-STAR does not \cmag{enforce} the block diagonal structure of the core tensor of the MSI (described in~\eqref{eq:tucker_mdl_msi_var_ext_core}). The following theorem gives a constructive proof of exact recovery conditions from which Alg.~\ref{alg:alg1} is derived.

\begin{theorem} \label{thm:exacr_rec_alg1}
Suppose that the HRI $\corange{\tensor{Z}_h}$ and the variability tensor $\tensor{\Psi}$ have multilinear ranks $(K_{Z,1},K_{Z,2},K_{Z,3})$ and $(K_{\Psi,1},K_{\Psi,2},K_{\Psi,3})$, respectively, that $\tensor{Y}_h$ and {\color{red}$\tensor{Y}_m$} admit Tucker decompositions as denoted in~\eqref{eq:alg1_coupled_eqs_formulation}, that the observation noise is zero (i.e. $\tensor{E}_h=\tensor{0}$, $\tensor{E}_m=\tensor{0}$), and that
\begin{align}
    & \rank(\bP_i \bB_{Z,i}) = K_{Z,i}\,, \,\,\, i\in\{1,2\}
    \label{eq:thm3_rank_cond1}
    \\
    & \rank(\bP_i \bB_{\Psi,i}) \leq K_{\Psi,i}\,, \,\,\, i\in\{1,2\}
    \label{eq:thm3_rank_cond2}
    \\
    & \rank(\unfold{\bY_{\!h}}{i}) = K_{Z,i}\,, \,\,\, i\in\{1,2,3\}
    \label{eq:thm3_rank_cond3}
    \\
    & \rank(\unfold{\bY_{\!m}}{i}) = K_{Z,i}+K_{\Psi,i} \leq N_i\,, \,\,\, i\in\{1,2\}
    \label{eq:thm3_rank_cond4}
\end{align}
Then, if all columns in $\bP_i \bB_{Z,i}$ are linearly independent from those in $\bP_i \bB_{\Psi,i}$, for $i\in\{1,2\}$, Algorithm~\ref{alg:alg1} exactly recovers $\corange{\tensor{Z}_h}$ from the observations.
\end{theorem}

\begin{proof}
Let us compute matrices $\widehat{\bC}_{m,i}$, $i\in\{1,2\}$, and $\widehat{\bC}_{h,i}$, $i\in\{1,2,3\}$ as the left-singular vectors associated with the non-zero singular values of $\unfold{\bY_{\!m}}{i}$, $i\in\{1,2\}$ and $\unfold{\bY_{\!h}}{i}$, $i\in\{1,2,3\}$, respectively. 
Then, due to~\eqref{eq:thm3_rank_cond3}--\eqref{eq:thm3_rank_cond4} and to the non-uniqueness of matrix decomposition, these matrices satisfy:
\begin{align}
    \widehat{\bC}_{h,i}  &= \bP_i \bB_{Z,i} \bQ_{h,i}, \,\,\, i\in\{1,2\}
    \label{eq:thm3_4a}
    \\
    \widehat{\bC}_{h,3}  &= \bB_{Z,3} \bQ_{h,3}
    \label{eq:thm3_4b}
    \\
    \widehat{\bC}_{m,i}  &= \bC_{m,i} \bQ_{m,i}, \,\,\, i\in\{1,2\}
    \label{eq:thm3_4c}
\end{align}
for invertible matrices $\bQ_{h,i}\in\amsmathbb{R}^{K_{Z,i}\times K_{Z,i}}$, $i\in\{1,2,3\}$ and $\bQ_{m,i}\in\amsmathbb{R}^{(K_{Z,i}+K_{\Psi,i})\times (K_{Z,i}+K_{\Psi,i})}$, $i\in\{1,2\}$.

Now, the main problem caused by the presence of variability is that matrices $\bQ_{m,1}$ and $\bQ_{m,2}$ preclude us from distinguishing the factors $\bB_{Z,1}$ and $\bB_{Z,2}$, associated with $\corange{\tensor{Z}_h}$, from $\bB_{\Psi,1}$ and $\bB_{\Psi,2}$, associated with $\tensor{\Psi}$, using only information available in the MSI. These two become mixed in the spatial factors $\widehat{\bC}_{m,i}$.
Nonetheless, consider the relationship between spatial degradation of the factors estimated from the MSI and the spatial factors of the HSI:
\begin{align}
    \bP_i \widehat{\bC}_{m,i} &= \bP_i \bC_{m,i} \bQ_{m,i}
    \label{eq:thm3_5a}
    \\
    &= \big[\bP_i \bB_{Z,i}, \, \bP_i \bB_{\Psi,i}\big] \bQ_{m,i} \,.
    \label{eq:thm3_5b}
\end{align}
Now, let us compute matrices $\widetilde{\bQ}_i\in\amsmathbb{R}^{(K_{Z,i}+K_{\Psi,i})\times K_{Z,i}}$, $i\in\{1,2\}$ such that
\begin{align}
    & \widehat{\bC}_{h,i} = \bP_i \widehat{\bC}_{m,i} \widetilde{\bQ}_i \,.
    \label{eq:thm3_6a}
\end{align}
By partitioning the following matrix product as $\bQ_{m,i}\widetilde{\bQ}_i=[\overline{\bQ}_{Z,i}^\top,\overline{\bQ}_{\Psi,i}^\top]^\top$,~\eqref{eq:thm3_6a} can be represented as
\begin{align}
    \widehat{\bC}_{h,i} &= \bP_i \bB_{Z,i} \bQ_{h,i} 
    \nonumber \\
    &= \bP_i \bB_{Z,i} \overline{\bQ}_{Z,i} + \bP_i \bB_{\Psi,i} \overline{\bQ}_{\Psi,i} \,.
    \label{eq:thm3_7a}
\end{align}

Since all columns in $\bP_i \bB_{Z,i}$ are linearly independent from those in $\bP_i \bB_{\Psi,i}$, equality~\eqref{eq:thm3_6a} will be satisfied if and only if the result of the product $\widehat{\bC}_{m,i} \widetilde{\bQ}_i$ (i.e., the right hand side of~\eqref{eq:thm3_7a}) does not contain any nontrivial linear combination of the columns of $\bP_i \bB_{\Psi,i}$. Thus, $\overline{\bQ}_{\Psi,i}=\cb{0}$ and $\overline{\bQ}_{Z,i}=\bQ_{h,i}$ due to~\eqref{eq:thm3_rank_cond1} and~\eqref{eq:thm3_rank_cond3}.
This allows us to ``separate'' the variability and image factors as
\begin{align}
    \widetilde{\bC}_{m,i} &= \widehat{\bC}_{m,i} \widetilde{\bQ}_i
    \nonumber\\
    & = \bB_{Z,i} \bQ_{h,i} \,,
\end{align}
for $i\in\{1,2\}$.
Now, consider the vectorization of the HSI as:
\begin{align}
    (\widehat{\bC}_{h,3} \otimes \bP_2\widetilde{\bC}_{m,2} \otimes \bP_1\widetilde{\bC}_{m,1} )\vect(\tensor{G}_Z) = \vect(\tensor{Y}_h) \,.
    \label{eq:thm3_7b}
\end{align}
Since {\color{magenta} the matrix in the  left hand side of~\eqref{eq:thm3_7b}  has full column rank}, $\widehat{\tensor{G}}_Z$ can be uniquely recovered from this equation, and will satisfy $\widehat{\tensor{G}}_Z=\tensor{G}_Z\times_1\bQ_{h,1}^{-1}\times_2\bQ_{h,2}^{-1}\times_3\bQ_{h,3}^{-1}$.
The HRI and the spectrally degraded scaling factors are then finally recovered as:
\begin{align}
\begin{split}
    \corange{\widehat{\tensor{Z}}_h} &= \ldbrack\widehat{\tensor{G}}_{Z}; \widetilde{\bC}_{m,1}, \widetilde{\bC}_{m,2}, \widehat{\bC}_{h,3}\rdbrack
    \\
    &= \ldbrack \tensor{G}_{Z}; \bB_{Z,1}, \bB_{Z,2}, \bB_{Z,3}\rdbrack = \tensor{Z}
\end{split}
\end{align}
and
\begin{align} \label{eq:recovering_low_psi1}
    \widehat{\tensor{\Psi}\times_3\bP_3}=\tensor{Y}_m-\corange{\widehat{\tensor{Z}}_h}\times_3\bP_3 \,,
\end{align}
which completes the proof.
\end{proof}


\corange{Note that CT-STAR does not use the spectral degradation operation $\bP_3$ to recover the HRI $\tensor{Z}_h$, what makes it a spectrally ``blind'' algorithm.}
The CT-STAR algorithm is \corange{also} fast (see Section~\ref{sec:comp}), but only works for the cases where the ranks of the spatial modes of $\corange{\tensor{Z}_h}$ are smaller than the dimensions of the HSI, which is quite restrictive. \corange{This can make CT-STAR unsuited to process real images which can have high rank values, and motivates the search for a method with more flexibility in the selection of the ranks.} Moreover, both Alg.~\ref{alg:alg1} and Theorem~\ref{thm:exacr_rec_alg1}, in considering model~\eqref{eq:alg1_coupled_eqs_formulation}, made no assumptions about the (block diagonal) structure of the core tensor of the MSI. Although this led to more freedom from a modeling perspective, the recoverability conditions turned out to be restrictive.
\corange{In the following section, we will explore the block diagonal structure of $\tensor{C}_m$ using an optimization-based algorithm to address these limitations.}

\section{An optimization-based algorithm}
\label{sec:alg2_optimizationBased}

In this section, we pursue a different approach. Assume that model~\eqref{eq:obs_mdl_hsi_lowrank}--\eqref{eq:obs_mdl_msi_lowrank_btd} holds and that the values forming the multilinear ranks of both $\corange{\tensor{Z}_h}$ and $\tensor{\Psi}$ are sufficiently low so that $\tensor{Y}_m$ admits a block term decomposition (BTD) in the noiseless case~\cite{lathauwer2008tensor_BTD2_uniqueness}. We can then use uniqueness results thereof to guarantee the identifiability of $\corange{\tensor{Z}_h}$ under less restrictive conditions.

Let us consider the image fusion problem as the solution to the following optimization problem:
\begin{align}
     \mathop{\min}_{\bTheta} & \,\, J(\bTheta) \triangleq \Big\|\tensor{Y}_h - \big\ldbrack {\tensor{G}}_{Z}; \bP_1 {\bB}_{Z,1}, \bP_2 {\bB}_{Z,2}, {\bB}_{Z,3} \big\rdbrack \Big\|_F^2
    \nonumber\\
    & \hspace{-1ex} + \corange{\lambda}\,\bigg\|\tensor{Y}_m - \sum_{\iota\in\{Z,\Psi\}} \big\ldbrack {\tensor{G}}_{\iota}; {\bB}_{\iota,1}, {\bB}_{\iota,2}, \bP_3 {\bB}_{\iota,3} \big\rdbrack \bigg\|_F^2
    \label{eq:opt_prob_fus_btd}
\end{align}
where $\bTheta=\{{\tensor{G}}_{\iota},{\bB}_{\iota,i}:\iota\in\{Z,\Psi\}, \, i\in\{1,2,3\}\}$ \corange{and $\lambda\in\amsmathbb{R}_+$ is a fixed parameter which balances the contribution of each term in the cost function. In the following, we first describe a procedure to solve problem~\eqref{eq:opt_prob_fus_btd} in Section~\ref{sec:CBSTAR_opt}, and later provide exact recovery guarantees in Section~\ref{sec:CBSTAR_recovery}.}

\begin{algorithm} [thb]
\footnotesize
\SetKwInOut{Input}{Input}
\SetKwInOut{Output}{Output}
\caption{\mbox{Optimization-based image fusion (CB-STAR)}\label{alg:alg2}}
\Input{\mbox{Images $\tensor{Y}_h$,$\tensor{Y}_m$ ranks $K_{Z,i}$,$K_{\Psi,i}$, $i\in\{1,2,3\}$\corange{, iterations~$F$}}}
\Output{HRI $\corange{\widehat{\tensor{Z}}_h}$, spectrally degraded variability factors $\opPspec(\widehat{\tensor{\Psi}})$}

Initialize $\bTheta^{(0)}$ according to Section~\ref{sec:alg2_initialization}\;

\While{Stopping criteria is not satisfied}{


\corange{Compute $\tensor{G}_Z$ and $\bB_{Z,i}$, $i\in\{1,2,3\}$ by solving~\eqref{eq:opt_subp_Z} with Algorithm~\ref{alg:alg3} in Appendix~\ref{sec:appendix1}, using $F$ iterations\;}

Compute $\bB_{\Psi,1}$, $\bB_{\Psi,2}$, $\bP_3\bB_{\Psi,3}$ and $\tensor{G}_{\Psi}$ by solving \eqref{eq:HOSVD_optPsi} using the high-order SVD with rank $(K_{\Psi,1},K_{\Psi,2},K_{\Psi,3})$\;
}
Compute $\corange{\widehat{\tensor{Z}}_h}=\ldbrack\widehat{\tensor{G}}_{Z}; \widehat{\bB}_{Z,1}, \widehat{\bB}_{Z,2}, \widehat{\bB}_{Z,3}\rdbrack$ \;

Compute $\opPspec(\widehat{\tensor{\Psi}})=\tensor{Y}_m-\corange{\widehat{\tensor{Z}}_h}\times_3\bP_3$ \;
\end{algorithm}

\subsection{Optimization}
\label{sec:CBSTAR_opt}



\corange{In order to minimize the cost function in~\eqref{eq:opt_prob_fus_btd}, we consider a block coordinate descent strategy, which successively minimizes $J$ w.r.t. $\tensor{Z}_h$ (i.e., $\tensor{G}_Z$ and $\bB_{Z,i}$, $i\in\{1,2,3\}$) and w.r.t. $\tensor{\Psi}$ (i.e., $\tensor{G}_{\Psi}$ and $\bB_{\Psi,i}$, $i\in\{1,2,3\}$), while keeping the remaining variables fixed.}
The optimization procedure is detailed in Alg.~\ref{alg:alg2}, which we call CB-STAR (Coupled Block term decompositions for hyperspectral Super-resoluTion with vARiability).

\begin{color}{orange}

\subsubsection{Optimizing w.r.t. $\tensor{Z}_h$}
The optimization problem w.r.t. $\tensor{Z}_h$ can be written as:
\begin{align} \label{eq:opt_subp_Z}
    \min_{\tensor{G}_Z,\bB_{Z,i}} \,\, 
    & \Big\|\tensor{Y}_h - \big\ldbrack {\tensor{G}}_{Z}; \bP_1{\bB}_{Z,1}, \bP_2{\bB}_{Z,2}, {\bB}_{Z,3}\big\rdbrack \Big\|_F^2
    \nonumber \\
    +\lambda & \Big\|\cmag{\tensor{Y}}_0 - \big\ldbrack {\tensor{G}}_{Z}; {\bB}_{Z,1}, {\bB}_{Z,2}, \bP_3{\bB}_{Z,3} \big\rdbrack \Big\|_F^2 \,,
\end{align}
where $\cmag{\tensor{Y}}_0=\tensor{Y}_m-\big\ldbrack {\tensor{G}}_{\Psi}; {\bB}_{\Psi,1}, {\bB}_{\Psi,2}, \bP_3{\bB}_{\Psi,3}\big\rdbrack$. This is a variability-free, Tucker-based image fusion problem. We propose to solve~\eqref{eq:opt_subp_Z} using a block coordinate descent strategy w.r.t.
$\tensor{G}_Z$ and $\bB_{Z,i}$, $i\in\{1,2,3\}$, with a small number of iterations~$F$. This procedure is detailed in Alg.~\ref{alg:alg3} and in Appendix~\ref{sec:appendix1}.
An approximate closed form solution to~\eqref{eq:opt_subp_Z} can also be computed efficiently using the SCOTT algorithm~\cite{prevost2020coupledTucker_hyperspectralSRR_TSP}.
\end{color}

\subsubsection{Optimizing w.r.t. $\tensor{\Psi}$}
This optimization problem can be written equivalently as
\begin{align} \label{eq:HOSVD_optPsi}
    \min_{\tensor{G}_{\Psi},\bB_{\Psi,i},\bX_2} \,\, \Big\|\cmag{\tensor{Y}}_1 - \big\ldbrack {\tensor{G}}_{\Psi}; {\bB}_{\Psi,1}, {\bB}_{\Psi,2}, \bX_2 \big\rdbrack \Big\|_F^2 \,,
\end{align}
where $\cmag{\tensor{Y}}_1=\tensor{Y}_{m}-\ldbrack {\tensor{G}}_{Z}; {\bB}_{Z,1}, {\bB}_{Z,2}, \bP_3 {\bB}_{Z,3}\rdbrack$ and $\bX_2=\bP_3 {\bB}_{\Psi,3}$. This problem can be solved by computing the high-order SVD of $\cmag{\tensor{Y}}_1$ with rank $(K_{\Psi,1},K_{\Psi,2},K_{\Psi,3})$~\cite{lathauwer2000multilinearHOSVD}. Note that problem~\eqref{eq:HOSVD_optPsi} only returns $\bX_2=\bP_3 {\bB}_{\Psi,3}$ instead of $\bB_{\Psi,3}$. This is not a problem since the variations of $\bB_{\Psi,3}$ in the nullspace of $\bP_3$ are not identifiable.



\subsubsection{Initialization}
\label{sec:alg2_initialization}

Since this optimization problem is non-convex, the choice of initialization can have a significant impact on the performance of the algorithm. This can be particularly prominent in this algorithm since the model considered in~\eqref{eq:tensor_obs_mdl_hsi}--\eqref{eq:tensor_obs_mdl_msi_add} allows for a significant amount of ambiguity.
Fortunately, for practical scenes, we can consider a simple strategy to provide a reasonably accurate initial guess.

In the noiseless case, the following relation is satisfied:
\begin{align} \label{eq:initialization_psi_1}
    \opPspat(\tensor{Y}_m)
    & = \opPspat\big(\opPspec(\corange{\tensor{Z}_h} + \tensor{\Psi})\big)
    \\ \nonumber
    & = \opPspec(\tensor{Y}_h) + \opPspat\big(\opPspec(\tensor{\Psi})\big) \,.
\end{align}
Thus, we can obtain a spatially and spectrally degraded version of $\tensor{\Psi}$ directly from the HSI and MSI simply as:
\begin{align}
    \widetilde{\tensor{\Psi}} 
    & = \opPspat(\tensor{Y}_m) - \opPspec(\tensor{Y}_h)
    \nonumber \\
    & = \opPspat\big(\opPspec(\tensor{\Psi})\big) \,.
    \label{eq:initialization_psi_2}
\end{align}
Then, if $\tensor{\Psi}$ is smooth, we can spatially upscale $\widetilde{\tensor{\Psi}}$ using some form of interpolation (e.g., bicubic), leading to $\opPspec(\tensor{\Psi}^{(0)})$.
Finally, we can use the Tucker decomposition of $\opPspec(\tensor{\Psi}^{(0)})$ to initialize $\tensor{G}_{\Psi}$, $\bB_{\Psi,i}$, $i\in\{1,2,3\}$. We call this the \emph{interpolation} initialization.

Another option is to try to invert~\eqref{eq:initialization_psi_2} using the pseudoinverse of operator $\opPspat$, which can be computed using properties of the tensor vectorization and Kronecker product:
\begin{align}
    \big(\opPspat\big)^{\dagger} = 
    (\cdot) \times_1 \bP_1^\dagger \times_2 \bP_2^\dagger \,,
\end{align}
where $\bX^{\dagger}$ denotes the pseudoinverse of $\bX$. the initialization can be then computed as $\opPspec(\tensor{\Psi}^{(0)})=\widetilde{\tensor{\Psi}}\times_1 \bP_1^\dagger \times_2 \bP_2^\dagger$. We call this the \emph{pseudoinverse} initialization.

The initialization of $\bB_{Z,i}$, $i\in\{1,2,3\}$ can then be performed as ${\bB}_{Z,3} = \tSVD_{K_{Z,3}}(\unfold{\bY_{\!h}}{3})$ and ${\bB}_{Z,i} = \tSVD_{K_{Z,i}}(\unfold{\bX}{i})$ for $i\in\{1,2\}$, where $\tensor{X}=\tensor{Y}_m-\opPspec(\tensor{\Psi}^{(0)})$.

\subsection{Exact Recovery}
\label{sec:CBSTAR_recovery}

Suppose that $K_{Z,i}=K_{\Psi,i}$, $i\in\{1,2,3\}$, without loss of generality, so that the MSI follows a standard BTD as considered in~\cite{lathauwer2008tensor_BTD2_uniqueness}. Then we have the following result regarding the identifiability of the proposed algorithm.

\begin{theorem} \label{thm:exacr_rec_alg2}
Suppose that $K_i\equiv K_{Z,i}=K_{\Psi,i}$, $i\in\{1,2,3\}$, that the observations are noise free (i.e., $\tensor{E}_h=\tensor{0}$, $\tensor{E}_m=\tensor{0}$), that $\{\tensor{G}_{\iota},\bB_{\iota,i}:\iota\in\{Z,\Psi\},\,i\in\{1,2,3\}\}$ are drawn from some joint absolutely continuous distribution, and that the following conditions on the dimensions hold:
\begin{align}
    M_1 \ge 2 K_1  & \mbox{ and } M_2 \ge 2 K_2
    \label{eq:dimension_conditions_thm_alg2_c}
    \\
    K_{Z,3} \leq \min & \big\{ N_1 N_2, \, K_{Z,1}K_{Z,2} \big\}
    \label{eq:dimension_conditions_thm_alg2_ef}
\end{align}
and either one of the following:
\begin{align}
\begin{cases}
    \begin{cases}
    K_3 > K_1+K_2-2 \,, &  \text{or} \\
    |K_1-K_2| > K_3-2 \,, & \\
    \end{cases}
    \\ 
    \mbox{ and } L_m \ge 2 K_3
\end{cases}
\label{eq:dimension_conditions_thm_alg2_a}
\end{align}
or
\begin{align}
\begin{split}
    K_1=K_2 \,,\,\,\, K_3 \geq 3 \,\,\,\mbox{and}\,\,\, K_3 < L_m \,,
\end{split}
\label{eq:dimension_conditions_thm_alg2_b}
\end{align}
is satisfied. Then, the solution to optimization problem~\eqref{eq:opt_prob_fus_btd} satisfies \corange{$\widehat{\tensor{Z}}_h=\tensor{Z}_h$} almost surely.
\end{theorem}

\begin{proof}
Since there is no additive noise, the optimal solution to optimization problem~\eqref{eq:opt_prob_fus_btd} will necessarily make both terms in the cost function equal to zero. This implies that
\begin{align}
    \tensor{Y}_h & = \ldbrack \widehat{\tensor{G}}_{Z}; \bP_1\widehat{\bB}_{Z,1}, \bP_2\widehat{\bB}_{Z,2}, \widehat{\bB}_{Z,3} \rdbrack \,,
    \\
    \tensor{Y}_m & = \sum_{\iota\in\{Z,\Psi\}} \ldbrack \widehat{\tensor{G}}_{\iota}; \widehat{\bB}_{\iota,1}, \widehat{\bB}_{\iota,2}, \bP_3 \widehat{\bB}_{\iota,3} \rdbrack \,,
    \label{eq:recovered_MSI_BTD_opt}
\end{align}
where $\{\widehat{\tensor{G}}_{\iota},\widehat{\bB}_{\iota,i}:\iota\in\{Z,\Psi\},\,i\in\{1,2,3\}\}$ denotes a solution to~\eqref{eq:opt_prob_fus_btd}.

Since the dimension conditions~\eqref{eq:dimension_conditions_thm_alg2_c} and either~\eqref{eq:dimension_conditions_thm_alg2_a} or~\eqref{eq:dimension_conditions_thm_alg2_b} are satisfied and the core tensor and factor matrices are drawn from joint absolutely continuous distributions, the BTD decomposition of the MSI in~\eqref{eq:recovered_MSI_BTD_opt} is essentially unique according to Theorems~5.1 and~5.5 in~\cite{lathauwer2008tensor_BTD2_uniqueness}. This means that the following conditions are satisfied:
\begin{align}
    \widehat{\bB}_{Z,i} & = \bB_{\iota_1,i} \bQ_{\iota_1,i}, \,\, i\in\{1,2\}
    \label{eq:ambig_thm2_1}
    \\
    \widehat{\bB}_{\Psi,i} & = \bB_{\iota_2,i} \bQ_{\iota_2,i}, \,\, i\in\{1,2\}
    \label{eq:ambig_thm2_2}
    \\
    \bP_3\widehat{\bB}_{Z,3} & = \bP_3 \bB_{\iota_1,3} \bQ_{\iota_1,3}
    \label{eq:ambig_thm2_3}
    \\
    \bP_3\widehat{\bB}_{\Psi,3} & = \bP_3 \bB_{\iota_2,3} \bQ_{\iota_2,3}
    \label{eq:ambig_thm2_4}
    \\
    \widehat{\tensor{G}}_{Z} & = \tensor{G}_{\iota_1} \times_1 \bQ_{\iota_1,1}^{-1} \times_2 \bQ_{\iota_1,2}^{-1} \times_3 \bQ_{\iota_1,3}^{-1}
    \label{eq:ambig_thm2_5}
    \\
    \widehat{\tensor{G}}_{\Psi} & = \tensor{G}_{\iota_2} \times_1 \bQ_{\iota_2,1}^{-1} \times_2 \bQ_{\iota_2,2}^{-1} \times_3 \bQ_{\iota_2,3}^{-1}
    \label{eq:ambig_thm2_6}
\end{align}
{\color{magenta} To account for the possible permutation of the two BTD terms,  indexes $\iota_1$ and $\iota_2$ can be either $(\iota_1,\iota_2)=(Z,\Psi)$ or $(\iota_1,\iota_2)=(\Psi,Z)$,} and $\bQ_{Z,i}$, $\bQ_{\Psi,i}$, $i\in\{1,2,3\}$ are invertible matrices of appropriate size, {\color{magenta} which account for rotational \corange{and scaling} ambiguities of the model.}

{\color{magenta} Let us consider the mode-3 unfolding of the spectrally degraded HSI:}
\begin{align*}
    \bP_3\unfold{{\bY_{\!h}}}{3} & = \bP_3\widehat{\bB}_{Z,3} \underbrace{\unfold{\mathop{\widehat{\bG}_{Z}}}{3} \big(\bP_2\widehat{\bB}_{Z,2} \otimes \bP_1\widehat{\bB}_{Z,1}\big)^\top}_{\,\widehat{\!\bX}} 
    \\
    & = \bP_3{\bB}_{Z,3} \underbrace{\unfold{\mathop{{\bG}_{Z}}}{3} \big(\bP_2{\bB}_{Z,2} \otimes \bP_1{\bB}_{Z,1}\big)^\top}_{\bX} \,,
\end{align*}
where $\,\widehat{\!\bX}$ and ${\bX}$ are generically full row rank by the assumption~\eqref{eq:dimension_conditions_thm_alg2_ef} on the ranks and on the dimensions. 
Therefore, we have
\begin{align}
    \vspan\big( \bP_3\widehat{\bB}_{Z,3} \big) = \vspan\big( \bP_3{\bB}_{Z,3} \big) \,.
    \label{eq:thm2_btd_span_eq}
\end{align}

However, since $K_3<L_m$ and due to the distributional assumptions on the factor matrices, we have that, generically, matrices $\bP_3\bB_{Z,3}$ and $\bP_3\bB_{\Psi,3}$ both have rank $K_3$, and matrix $\begin{bmatrix}
\bP_3 {\bB}_{Z,3}  &  \bP_3\bB_{\Psi,3} 
\end{bmatrix}$ is full column rank (i.e., it has rank greater than $K_3$). Therefore, the subspaces spanned by $\bP_3\bB_{Z,3}$ and $\bP_3\bB_{\Psi,3}$ are different, and it is not possible to have $\bP_3 {\bB}_{Z,3}  =  \bP_3\bB_{\Psi,3} \bS$, for any matrix $\bS$.
Thus,~\eqref{eq:thm2_btd_span_eq} implies that the equation $\bP_3\widehat{\bB}_{Z,3}=\bP_3\bB_{\Psi,3}\bQ_{\Psi,3}$ is not possible (i.e., it cannot be satisfied for any matrix $\bQ_{\Psi,3}$). This ensures, due to \eqref{eq:ambig_thm2_3} (and to the essential uniqueness of the MSI BTD), that we must have $\iota_1=Z$ in~\eqref{eq:ambig_thm2_1},~\eqref{eq:ambig_thm2_3} and~\eqref{eq:ambig_thm2_5}, (i.e., the factor matrices related to $\tensor{\Psi}$ can not fit the HSI), which shows that the correct permutation of the BTD terms is selected.

Finally, since $\,\widehat{\!\bX}$ is generically full row rank, (or since $\unfold{{\bY_{\!h}}}{3}$ has rank $K_3$), we have that $\widehat{\bB}_{Z,3}=\bB_{Z,3}\bS$ for some $\bS$, and~\eqref{eq:ambig_thm2_3} ensures that $\bS=\bQ_{Z,3}$. This means that $\widehat{\bB}_{Z,3}=\bB_{Z,3}\bQ_{Z,3}$ (almost surely), and, using~\eqref{eq:ambig_thm2_1} and~\eqref{eq:ambig_thm2_5} and $\iota_1=Z$, the reconstructed image consequently satisfies
\begin{align}
    \corange{\widehat{\tensor{Z}}_h} & = \big\ldbrack \widehat{\tensor{G}}_{Z}; \widehat{\bB}_{Z,1}, \widehat{\bB}_{Z,2}, \widehat{\bB}_{Z,3} \big\rdbrack 
    \nonumber \\
    & = \big\ldbrack {\tensor{G}}_{Z}; {\bB}_{Z,1}, {\bB}_{Z,2}, {\bB}_{Z,3} \big\rdbrack 
    = \corange{\tensor{Z}_h},
\end{align}
(almost surely), which concludes the proof.
\end{proof}

By taking the block diagonal structure of $\tensor{C}_m$ into account, Theorem~\ref{thm:exacr_rec_alg2} obtains generally less restrictive recovery conditions. Comparing Theorems~\ref{thm:exacr_rec_alg1} and~\ref{thm:exacr_rec_alg2}, we can see that: 1) conditions~\eqref{eq:thm3_rank_cond3} and~\eqref{eq:dimension_conditions_thm_alg2_ef} are equivalent; 2) the conditions for the spectral ranks are not directly comparable but are similarly restrictive for both theorems; and, most notably, 3) The constraint~\eqref{eq:thm3_rank_cond4} on spatial ranks is much more restrictive when compared to the one in~\eqref{eq:dimension_conditions_thm_alg2_c}, required by Theorem~\ref{thm:exacr_rec_alg2}. \corange{This shows that, although computationally more demanding, CB-STAR has more flexibility and may be able to deliver \cdgreen{better performance} when compared to CT-STAR for images with complex spatial content.}

\section{\corange{Computational complexity}}\label{sec:comp}

\corange{The computational complexity of the algorithms is given as follows. The total cost involved with the main operations in CT-STAR (Alg.~\ref{alg:alg1}) are the following: 1) computation of the truncated SVDs in steps~2 and~3, which requires $\mathcal{O}(\max\{K_{Z,1},K_{Z,2}\}M_1M_2L_m+K_{Z,3}N_1N_2L_h)$ flops, 2) computation of the equations in steps~4 and~5, which requires~$\mathcal{O}(\max\{K_{Z,1},K_{Z,2}\}N_1N_2L_h+M_1N_1K_{Z,1}+M_2N_2K_{Z,2})$ flops, and solution of the linear equation in step~6, which requires~$\mathcal{O}\big(L_hN_1N_2(K_{Z,1}K_{Z,2}K_{Z,3})^2\big)$ flops.}

\corange{ For CB-STAR (Alg.~\ref{alg:alg2}), the total cost involved in each iteration is 1) solving problem~\eqref{eq:opt_subp_Z} in step~3 using Alg.~\ref{alg:alg3}, whose main costs are based on the solution to Sylvester equations, and are given by $\mathcal{O}\big(F(M_1^3+M_2^3+L_h^3+(K_{Z,2}K_{Z,3})^3)\big)$ (where we assume $K_{Z,1}\geq K_{Z,3}$ for simplicity), and 2) computing the high-order SVD in step~4, which costs~$\mathcal{O}\big(\max_i\{K_{\Psi,i}\}M_1M_2L_m\big)$ flops.
}


\section{Experiments}
\label{sec:results}

In this section, the performance of the proposed approach is illustrated through numerical experiments considering both synthetic and real data containing spatial and spectral variability.
All simulations were coded in MATLAB and run on a desktop with a 4.2 GHz Intel Core i7 and 16GB RAM.

\subsection{Experimental Setup} \label{sec:sim_setup}

We compared CT-STAR and CB-STAR to both matrix and tensor factorization-based algorithms. Among the matrix factorization-based methods, we considered the HySure~\cite{simoes2015HySure} and CNMF~\cite{yokoya2012coupledNMF} methods, the FuVar~\cite{Borsoi_2018_Fusion} method, which accounts for spectral variability, and the multiresolution analysis-based GLP-HS algorithm~\cite{aiazzi2006GLP_HS}. We also considered the \corange{LTMR~\cite{dian2019hyperspectralSRR_tensorMultiRank},} STEREO~\cite{kanatsoulis2018hyperspectralSRR_coupledCPD} and SCOTT~\cite{prevost2020coupledTucker_hyperspectralSRR_TSP} algorithms, \corange{which are tensor-based image fusion methods.}

The real \corange{HRIs and MSIs}, which were acquired at different time instants but at the same spatial resolution, were pre-processed as described in~\cite{simoes2015HySure}. \corange{This consisted} in the manual removal of water absorption and low-SNR bands, followed by the normalization of all bands of the \corange{HRIs and MSIs} such that the 0.999 intensity quantile \corange{corresponded} to a value of~1. 
Afterwards, the \corange{HRIs were} denoised (as described in~\cite{roger1996denoisingHSI}) to yield the high-SNR reference image~$\corange{\tensor{Z}_h}$~\cite{yokoya2017HS_MS_fusinoRev}.
The observed HSIs $\tensor{Y}_{h}$ were then generated from $\corange{\tensor{Z}_h}$ by applying a separable degradation operator, with $\bP_1=\bP_2$ (a Gaussian filter with unity variance followed by a subsampling with a decimation factor of two\footnote{Details on how to construct $\bP_1$ and $\bP_2$ can be found in~\cite{prevost2020coupledTucker_hyperspectralSRR_TSP}.}). Gaussian noise was also added to obtain an SNR of 30dB.
The observed MSIs $\tensor{Y}_{m}$ were generated by adding noise to the reference MSI to obtain an SNR of 40dB. The spectral response function $\bP_3$ was obtained from calibration measurements and known a priori\footnote{Available for download \href{https://earth.esa.int/web/sentinel/user-guides/sentinel-2-msi/document-library/-/asset_publisher/Wk0TKajiISaR/content/sentinel-2a-spectral-responses}{here}.}.
\corange{Note that the HRI $\tensor{Z}_h$ is not known by the algorithms, and is just used to assess their performance during the experiments using quantitative metrics.}

The parameters of the algorithms were selected as follows. 
We selected the ranks and regularization parameters for HySure, CNMF, FuVar and LTMR according to the original works~\cite{simoes2015HySure,yokoya2012coupledNMF,Borsoi_2018_Fusion,dian2019hyperspectralSRR_tensorMultiRank}. For STEREO, we selected the rank in the interval $[5,80]$ which led to the best reconstruction results. Similarly, for SCOTT and for the proposed algorithms, we selected the spatial ranks in the intervals $[10,80]$ and the spectral ranks in the interval $[2,30]$, which led to the best reconstruction results. \corange{For simplicity, we also set $\lambda=1$ for CB-STAR.} The spatial and spectral degradation operators (or, equivalently, the blurring kernels for HySure, CNMF and FuVar) were assumed to be known a priori for all methods.
The BCD procedure in Alg.~\ref{alg:alg2} was performed until the relative change in the objective function value was smaller than~$10^{-3}$. \corange{At each iteration of Alg.~\ref{alg:alg2}, \cdgreen{we ran Alg.~\ref{alg:alg3} for} one single inner iteration (i.e., $F=1$), which resulted in good experimental performance with moderate execution times.} Both the \emph{interpolation} and the \emph{pseudoinverse} initializations described in Section~\ref{sec:alg2_initialization} were considered, but only the first one (which performed better) is shown in the visual results.

To evaluate the quality of the reconstructed images $\corange{\widehat{\tensor{Z}}_h}$, we considered four quantitative metrics, which were previously used in~\mbox{\cite{yokoya2017HS_MS_fusinoRev,simoes2015HySure,Borsoi_2018_Fusion}}.
The first metric is the peak signal to noise ratio (PSNR), defined as
\begin{align}
	{\rm PSNR}
	{}={} \frac{10}{L} \sum_{\ell=1}^L
    \log_{10} \Bigg(
    \frac{M_1 M_2 \Ex\!\big\{\!\max\big(\corange{[\tensor{Z}_h]_{:,:,\ell}}\big)\big\}}{\big\|\corange{[\tensor{Z}_h]_{:,:,\ell}  -[\widehat{\tensor{Z}}_h]_{:,:,\ell}\big\|_F^2}} \Bigg)
    \nonumber\,,
\end{align}
where $\Ex\{\cdot\}$ denotes the expectation operator.

The second metric is the Spectral Angle Mapper (SAM):
\begin{align}
	{\rm SAM}
	{}={} \frac{1}{M_1 M_2} \sum_{n,m} \arccos \Bigg( \frac{\corange{[\tensor{Z}_h]_{:,n,m}^\top[\widehat{\tensor{Z}}_h]_{:,n,m}}}
    {\big\|\corange{[\tensor{Z}_h]_{:,n,m}}\big\|_2\big\|\corange{[\widehat{\tensor{Z}}_h]_{:,n,m}}\big\|_2} \Bigg)
    \nonumber\,.
\end{align}

The ERGAS~\cite{wald2000qualityERGAS} metric provides a global statistical measure of the quality of the fused data, and is defined as:
\begin{align}
	{\rm ERGAS}
	{}={} \frac{M_1M_2}{N_1N_2} \sqrt{\!\frac{10^4}{L_h} \!
    \sum_{\ell=1}^{L_h}  \frac{\big\|\corange{[\tensor{Z}_h]_{:,:,\ell}
    -[\widehat{\tensor{Z}}_h]_{:,:,\ell}}\big\|_F^2}{(\cb{1}^\top\corange{[\tensor{Z}_h]_{:,:,\ell}}\cb{1}/(M_1M_2))^2}} 
    \nonumber \,.
\end{align}

The last metric is the average of the bandwise UIQI~\cite{wang2002qualityUIQI}, which evaluates image distortions caused by loss of correlation and by luminance and contrast distortion, with value approaching one as $\corange{\widehat{\tensor{Z}}_h}$ approaches $\corange{\tensor{Z}_h}$.

We also evaluate the reconstructed images visually, by displaying true- and pseudo-color representations of the visual and infrared spectra of~$\corange{\widehat{\tensor{Z}}_h}$ (corresponding to the wavelengths $0.45$, $0.56$ and $0.66~\mu m$, and $0.80$, $1.50$ and~$2.20~\mu m$, respectively). Due to space limitations, we only display the results of FuVar, \corange{LTMR,} STEREO, SCOTT, CT-STAR and CB-STAR, since these are the methods which performed best, and the ones which were conceptually closest to our approach.
The spectrally degraded additive factors $\opPspec(\widehat{\tensor{\Psi}})$ estimated by CT-STAR and CB-STAR are also evaluated visually, through pseudo-color representations of its visible and infrared spectra, and by the norm (over all bands) of each of its pixels.


\begin{table}[ht!]	
\caption{Results - synthetic example (``rk'' stands for ``ranks'')} \vspace{-5pt}
\centering	
\renewcommand{\arraystretch}{1.15}
\setlength\tabcolsep{3.5pt}
\resizebox{\linewidth}{!}{
\begin{filecontents*}{tables/tab_synthetic_MC.txt}
    HySure	6.126785e-01	1.964419e+01	1.085478e+01	6.058010e-01	5.794876e+00
    CNMF	9.201972e-01	9.033478e+00	2.151924e+01	7.423314e-01	1.629600e+01
    GLPHS	7.684864e-01	4.540360e+00	2.744148e+01	9.277798e-01	7.085992e+00
    FuVar	9.779807e-01	7.846497e+00	2.268064e+01	8.242816e-01	1.342e+02
    LTMR	5.692345e+00	1.229166e+01	1.834960e+01	7.060704e-01	3.051048e+01
    STEREO	2.591059e+00	1.200981e+01	1.844402e+01	7.161172e-01	2.883092e-01
    SCOTT	5.809572e-01	7.418383e+00	2.218888e+01	8.930048e-01	3.761590e-02
    CT-STAR, rk:$(3,3,2),(2,2,1)$	5.890296e-01	7.254557e-01	4.388996e+01	9.984102e-01	4.962441e-02
    CT-STAR, rk:$(5,5,3),(3,3,2)$	5.388567e-01	6.346804e-01	4.512371e+01	9.987864e-01	7.440869e-02
    CT-STAR, rk:$(10,10,5),(5,5,3)$	5.036333e-01	5.903814e-01	4.565997e+01	9.989446e-01	1.058437e-01
    CT-STAR, rk:$(20,20,7),(7,7,3)$	1.193372e+00	1.286191e+00	3.922335e+01	9.950082e-01	2.480542e-01
    CB-STAR, rk:$(3,3,2),(2,2,1)$	6.181608e-01	7.778638e-01	4.335421e+01	9.981416e-01	2.564208e-01
    CB-STAR, rk:$(5,5,3),(3,3,2)$	5.399036e-01	6.274283e-01	4.533771e+01	9.987983e-01	2.756210e-01
    CB-STAR, rk:$(10,10,5),(5,5,3)$	4.971481e-01	5.487332e-01	4.657959e+01	9.990804e-01	4.180569e-01
    CB-STAR, rk:$(20,20,7),(7,7,3)$	1.399942e+00	1.466264e+00	3.845535e+01	9.935123e-01	1.084098e+00
\end{filecontents*}
\pgfplotstabletypeset[header=false,
col sep = tab,
	columns/0/.style={string type},
every head row/.style={before row={\hline}, after row=\hline},
columns/0/.style={string type, column name={Algorithm}, column type/.add={@{}@{\,}}{}},
columns/1/.style={column name={SAM}, column type/.add={@{\,}|@{\,}}{}},
columns/2/.style={column name={ERGAS}, column type/.add={@{\,}|@{\,}}{}},
columns/3/.style={column name={PSNR}, column type/.add={@{\,}|@{\,}}{}},
columns/4/.style={column name={UIQI}, column type/.add={@{\,}|@{\,}}{}},
columns/5/.style={column name={time}, fixed, precision=2, column type/.add={@{\,}|@{\,}}{}},
columns ={0,1,2,3,4,5},
every row 9 column 1/.style={postproc cell content/.style={ @cell content/.add={$\bf}{$} }},
every row 13 column 1/.style={postproc cell content/.style={ @cell content/.add={$\bf}{$} }},
every row 13 column 2/.style={postproc cell content/.style={ @cell content/.add={$\bf}{$} }},
every row 13 column 3/.style={postproc cell content/.style={ @cell content/.add={$\bf}{$} }},
every row 7 column 4/.style={postproc cell content/.style={ @cell content/.add={$\bf}{$} }},
every row 8 column 4/.style={postproc cell content/.style={ @cell content/.add={$\bf}{$} }},
every row 9 column 4/.style={postproc cell content/.style={ @cell content/.add={$\bf}{$} }},
every row 10 column 4/.style={postproc cell content/.style={ @cell content/.add={$\bf}{$} }},
every row 11 column 4/.style={postproc cell content/.style={ @cell content/.add={$\bf}{$} }},
every row 12 column 4/.style={postproc cell content/.style={ @cell content/.add={$\bf}{$} }},
every row 13 column 4/.style={postproc cell content/.style={ @cell content/.add={$\bf}{$} }},
every row 6 column 5/.style={postproc cell content/.style={ @cell content/.add={$\bf}{$} }},
]{tables/tab_synthetic_MC.txt}}
\label{tab:ex0_bis}
\end{table}

\subsection{Examples -- Synthetic data}

To evaluate the \corange{proposed} algorithms in a controlled scenario, we first considered a simulation with synthetic data.
The tensors $\corange{\tensor{Z}_h}$ and $\tensor{\Psi}$, of dimensions $100\times100\times200$, were generated following the Tucker model, with uniformly distributed entries on the interval $[0,1]$ and ranks $(10,10,5)$ and $(5,5,3)$, respectively. 
The spectral response function $\bP_3\in\amsmathbb{R}^{10\times200}$ was constructed by uniformly averaging groups of $20$ bands, and the rest of the simulation setup was the same as described in Section~\ref{sec:sim_setup}. For this experiment, we initialized CB-STAR with the results of CT-STAR.
\corange{\cdgreen{We considered} two examples. First, \cdgreen{we compared} the proposed method to other state-of-the-art algorithms, for different choices of rank. Afterwards, \cdgreen{we evaluated} the sensitivity of CT-STAR and CB-STAR to the presence of additive noise. In both cases, we report average results of a Monte Carlo simulation with 100 \cmag{noise} realizations.}

\subsubsection{Comparison to other algorithms}

\corange{For this comparison, we set the} ranks of STEREO and SCOTT were $50$ and $(60,60,5)$, respectively. \corange{\cdgreen{We also ran} the proposed methods with four different rank values (smaller, equal, and larger than the true data ranks)}, indicated in Table~\ref{tab:ex0_bis}.
The results in Table~\ref{tab:ex0_bis} show that the proposed methods yielded significant improvements when compared to the other algorithms, which is expected since this dataset was generated according to the model~\eqref{eq:tensor_obs_mdl_hsi}--\eqref{eq:tensor_obs_mdl_msi_add}. Moreover, the performance of both CT-STAR and CB-STAR as a function of the ranks was similar, with the best results \corange{observed} when the rank was the same as the ground truth, but with similar performance when the rank was underestimated. \corange{This indicates that CT-STAR and CB-STAR can still perform well when the data rank is higher than the one specified for the model.} When the selected rank was overestimated, the performance of the proposed methods degraded more sharply (with a more prominent decrease for CB-STAR), indicating that the ranks should not be much greater than the true values in order to obtain the best performance.

\begin{color}{orange}
\subsubsection{Effect of noise} 
The exact recovery results obtained in Theorems~\ref{thm:exacr_rec_alg1} and~\ref{thm:exacr_rec_alg2} assume a noiseless observation model, which is not the case in real applications.
To illustrate how the performance of the proposed methods is affected by the presence of additive noise~$\tensor{E}_h$ and~$\tensor{E}_m$, \cdgreen{we evaluated} the quantitative reconstruction metrics for different SNRs, varying from 0dB up to the noiseless case (SNR = $\infty$). For simplicity, we set $K_Z$ and $K_{\Psi}$ as the correct ranks of $\tensor{Z}_h$ and $\tensor{\Psi}$, respectively. The results are shown in Table~\ref{tab:exp_synth_SNRs}. It can be seen that the performance of both CT-STAR and CB-STAR decreased for low SNRs. Moreover, CB-STAR achieved better results for high SNRs ($\geq$40dB), but was surpassed by CT-STAR for low SNR scenarios ($\leq$30dB), whose results were considerably better for extremely noisy scenarios (0dB). Moreover, for SNRs found in typical scenes ($\geq$20dB), \cdgreen{these results indicate that} both methods are able to obtain satisfactory results.
Finally, we note that since the ranks of $\tensor{Z}_h$ and $\tensor{\Psi}$ \cdgreen{satisfied the} requirements of Theorems~\ref{thm:exacr_rec_alg1} and~\ref{thm:exacr_rec_alg2}, exact reconstruction was obtained for SNR~=~$\infty$ (i.e., the original and reconstructed images were equal up to machine precision).

The theoretical results presented in this work \cdgreen{made use of two hypotheses} to guarantee the recovery of the HRI, namely, that the image and variability factors have low multilinear rank, and that there is no observation noise. However, those hypotheses are often not strictly satisfied in practice. Nonetheless, CT-STAR and CB-STAR only perform a low-rank approximation of the data. Thus, they can still perform well when the true rank of the image is larger than the one specified to the algorithm and under the presence of noise, as illustrated in the previous simulations (even though exact recovery is not guaranteed in those cases). In the following section, we will evaluate the performance of the proposed methods with real HSIs and MSIs, which do not necessarily come from the specified low-rank models.

\end{color}

\begin{table}[th]
\renewcommand{\arraystretch}{1.15}
\setlength\tabcolsep{3.5pt}
\centering
\caption{\corange{Performance of the proposed methods for different SNRs}}  \vspace{-5pt}
\resizebox{\linewidth}{!}{
\begin{tabular}{c|cccccccccc} \hline
\multicolumn{10}{c}{CT-STAR} \\\hline
SNR	&	$0$	&	$10$	&	$20$	&	$30$	&	$40$	&	$60$	&	$80$	&	$100$	&	$\infty$	\\\hline
SAM	&	$9.622$	&	$3.191$	&	$1.197$	&	$0.647$	&	$0.370$	&	$0.033$	&	$0.003$	&	$0$	&	$0$	\\
ERGAS	&	$15.33$	&	$6.948$	&	$2.005$	&	$0.887$	&	$0.436$	&	$0.040$	&	$0.004$	&	$0$	&	$0$	\\
PSNR	&	$17.18$	&	$23.76$	&	$34.59$	&	$41.72$	&	$48.35$	&	$69.23$	&	$89.35$	&	$109.4$	&	$297.4$	\\
UIQI	&	$0.501$	&	$0.856$	&	$0.987$	&	$0.998$	&	$0.999$	&	$1$	&	$1$	&	$1$	&	$1$	\\
\hline
\multicolumn{10}{c}{CB-STAR} \\\hline
SNR	&	$0$	&	$10$	&	$20$	&	$30$	&	$40$	&	$60$	&	$80$	&	$100$	&	$\infty$	\\\hline
SAM	&	$36.89$	&	$14.47$	&	$4.567$	&	$1.272$	&	$0.295$	&	$0.023$	&	$0.002$	&	$0$	&	$0$	\\
ERGAS	&	$52.83$	&	$17.26$	&	$5.130$	&	$1.433$	&	$0.342$	&	$0.028$	&	$0.003$	&	$0$	&	$0$	\\
PSNR	&	$7.74$	&	$17.32$	&	$27.66$	&	$38.17$	&	$50.54$	&	$71.99$	&	$91.95$	&	$111.9$	&	$265.9$	\\
UIQI	&	$0.157$	&	$0.590$	&	$0.927$	&	$0.994$	&	$1$	&	$1$	&	$1$	&	$1$	&	$1$	\\
\hline
\end{tabular}}
\label{tab:exp_synth_SNRs}
\end{table}

\begin{figure}[thb]
	\centering
    \begin{minipage}{0.36\linewidth}
        \centering
        \includegraphics[width=1\linewidth]{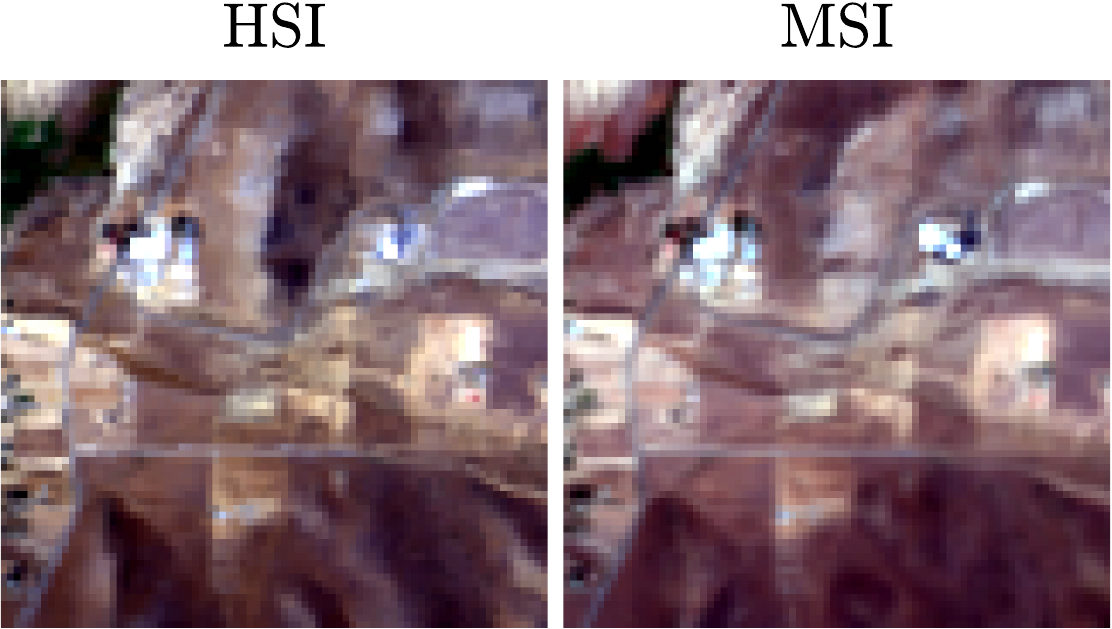}
        Lake Isabella
    \end{minipage}
    \,\,
    \begin{minipage}{0.42\linewidth}
        \centering
        \includegraphics[width=1\linewidth]{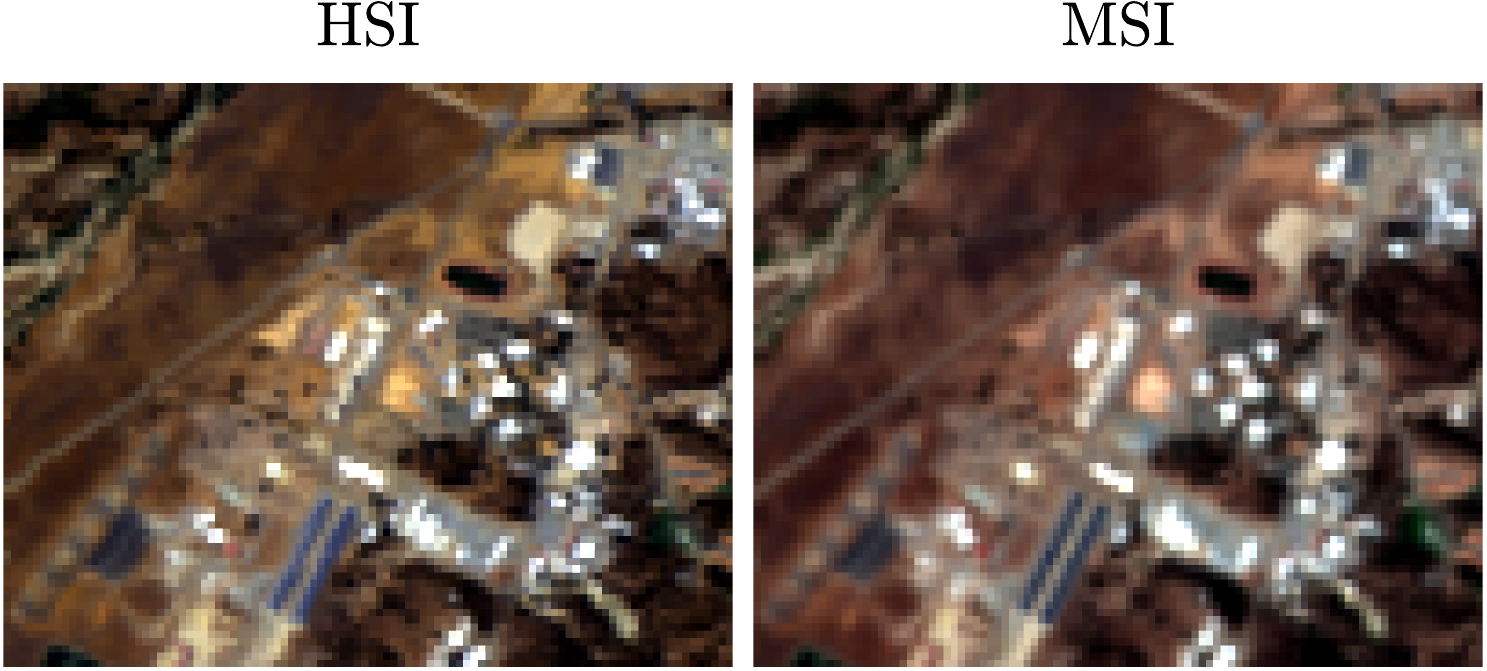}
        Lockwood
    \end{minipage}
    \caption{Hyperspectral and multispectral images with a small acquisition time difference used in the experiments.}
    \label{fig:ex3_HSIs_MSIs_a}
\end{figure}
\begin{figure}[thb]
	\centering
    \begin{minipage}{0.28\linewidth}
        \centering
        \includegraphics[width=1\linewidth]{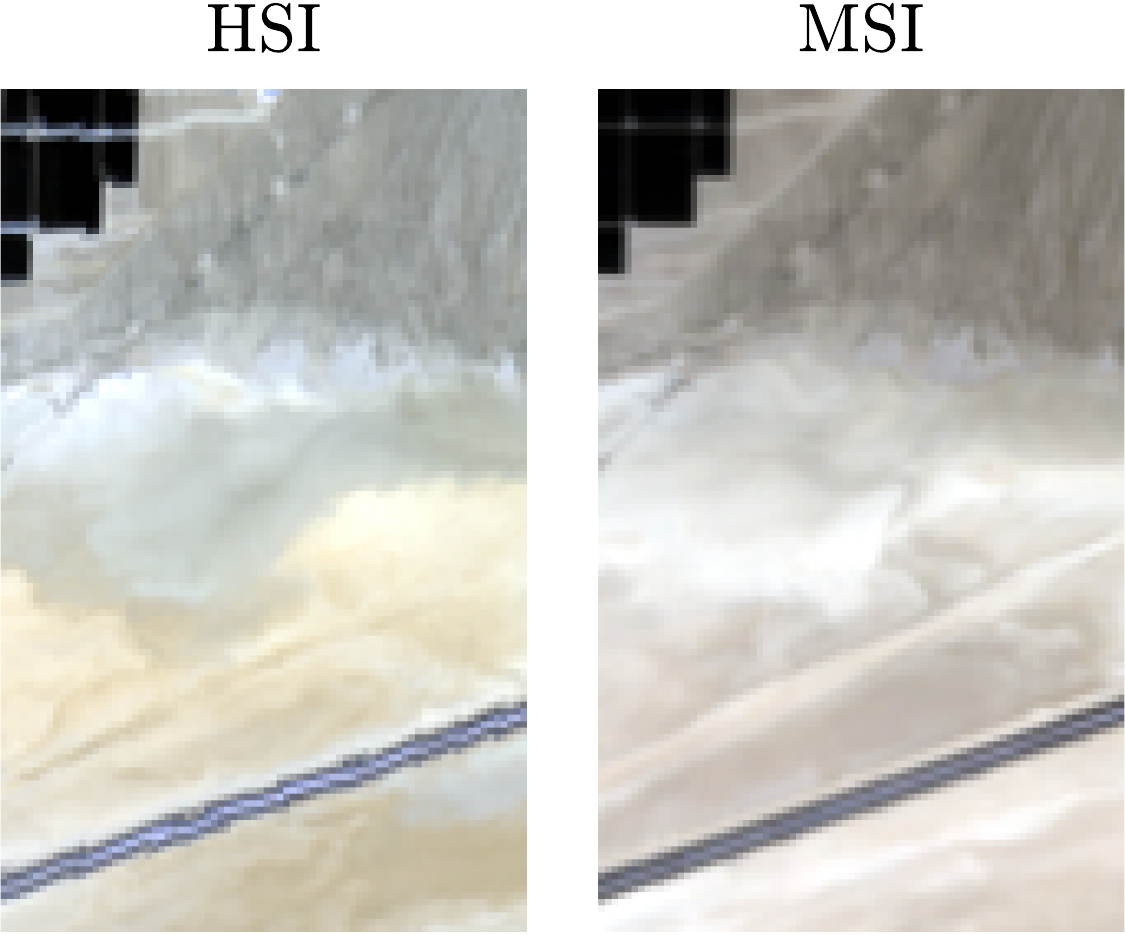}
        Ivanpah Playa
    \end{minipage}
    \,\,
    \begin{minipage}{0.32\linewidth}
        \centering
        \includegraphics[width=1\linewidth]{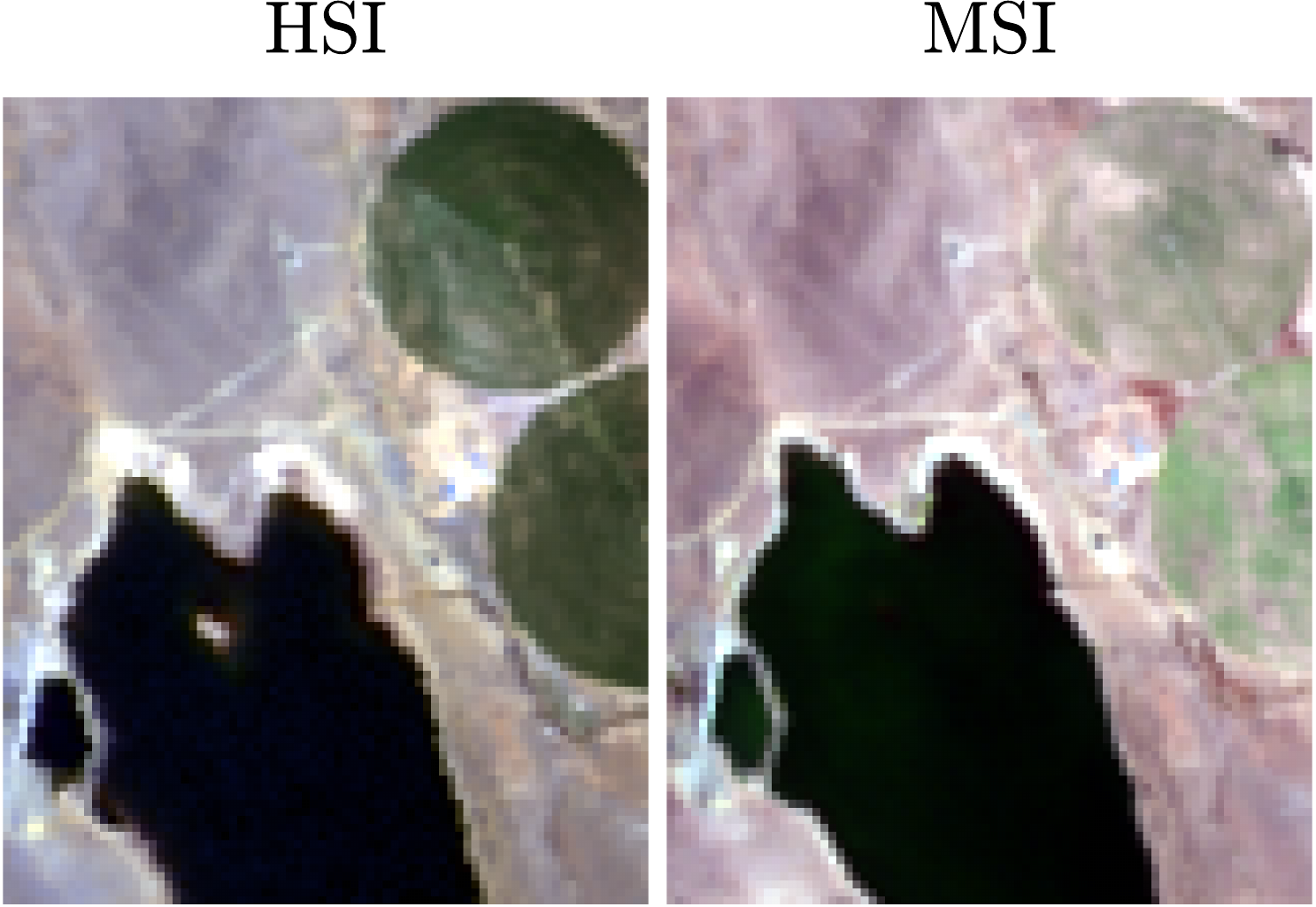}
        Lake Tahoe A
    \end{minipage}
    \,\,
    \begin{minipage}{0.32\linewidth}
        \centering
        \includegraphics[width=1\linewidth]{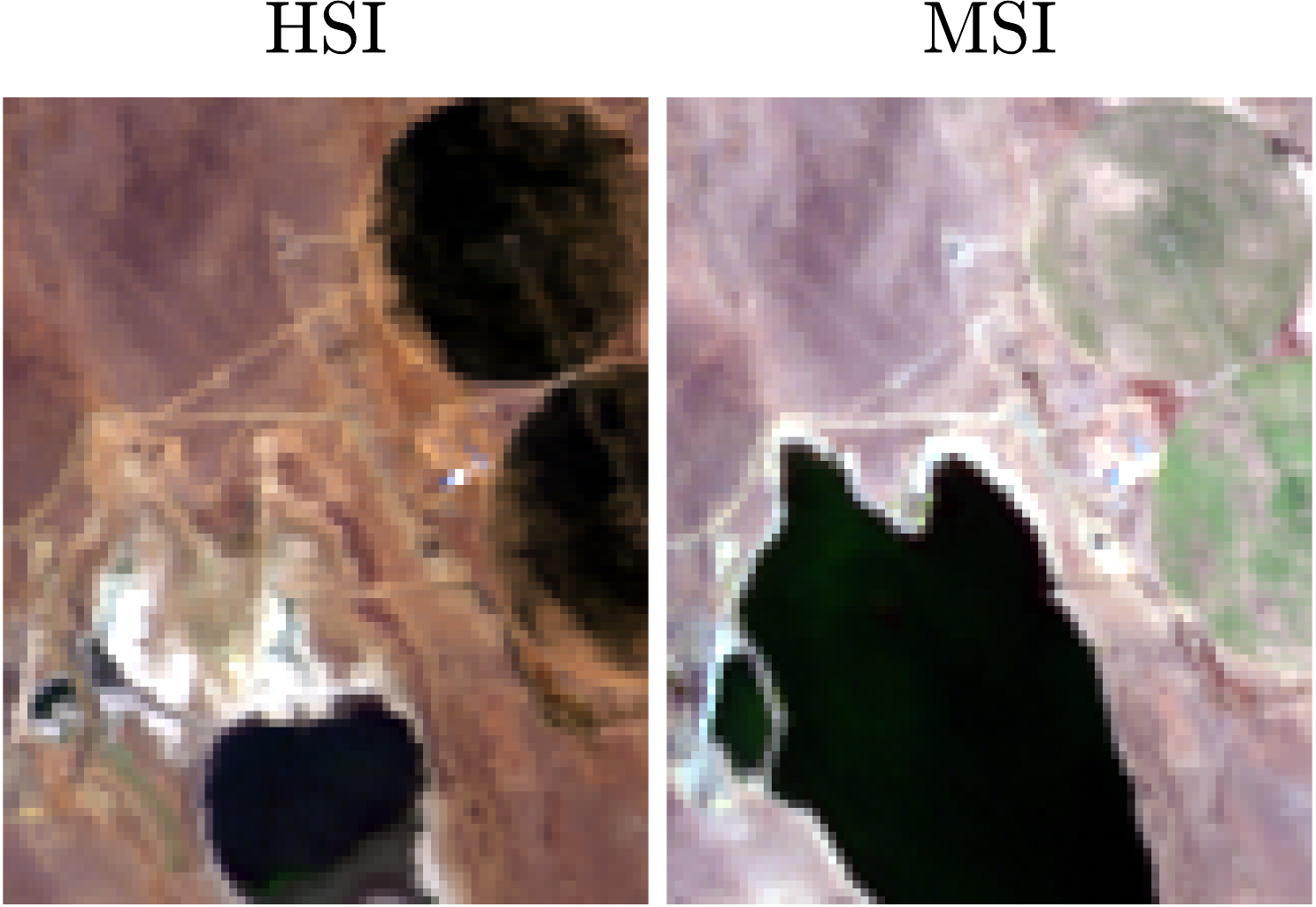}
        Lake Tahoe B
    \end{minipage}
    \caption{Hyperspectral and multispectral images with a large acquisition time difference used in the experiments.}
    \label{fig:ex3_HSIs_MSIs_b}
\end{figure}

\subsection{Example -- Real data}

In this example, we evaluated the algorithms using real HS and MS images acquired at different time instants, thus presenting different acquisition and seasonal conditions.
The reference hyperspectral and multispectral images, with a pixel size of 20~m, acquired by the AVIRIS and by the Sentinel-2A instruments, respectively, were originally considered in~\cite{Borsoi_2018_Fusion}. 
Four sets of image pairs were available. Two of which contained images acquired less than three months apart (thus containing moderate variability). The other two contained images acquired with a time difference of more than one year (thus containing more significant variability).
The HSI and MSI contained $L_h=173$ and $L_m=10$ bands, respectively. 
The selected ranks for the tensor-based methods are shown in Table~\ref{tab:exp_ranks_values}.
\corange{Although the ranks of CB-STAR \cdgreen{satisfied Theorem~\ref{thm:exacr_rec_alg2} only for} the Ivanpah Playa image pair, \cdgreen{this did not have} a negative impact on its performance, as will be shown in the following.}


\subsubsection{\cmag{Rank} \corange{sensitivity analysis}}
\corange{
%
Before proceeding, we evaluate how sensitive the performance of the proposed methods is to the selection of the ranks by plotting the PSNR as a function of each of the ranks $K_{Z,i}$ and $K_{\Psi,i}$, while keeping the others fixed. 
For simplicity, we kept the spatial ranks equal to each other (i.e., $K_{Z,1}=K_{Z,2}$ and $K_{\Psi,1}=K_{\Psi,2}$), and only show the results for the Lockwood image (to be described in Section~\ref{sec:res_realdata_moderate}) due to space limitations\footnote{Additional results are available on the supplemental material.}. The results, shown in Fig.~\ref{fig:sensitivity_ex6}, indicate that CT-STAR performs well when $K_{\Psi,1}=K_{\Psi,2}$ are small, and performs well for values of $K_{Z,3}$ which are not small. The optimal $K_{Z,1}=K_{Z,2}$ were relatively large, but a drop in performance was observed when they approach their upper limit $N_i-K_{\Psi,i}$, $i\in\{1,2,\}$. The performance of CB-STAR increased steadily with $K_{Z,1}=K_{Z,2}$ and for small values of $K_{Z,3}$, but decreased more sharply when all values $K_{Z,i}$, $i\in\{1,2,3\}$ were large. The variability ranks $K_{\Psi,i}$, on the other hand, did not affect the results too much when $K_{\Psi,1}=K_{\Psi,2}$ were sufficiently large. 
}

\begin{figure}
    \centering
    \begin{minipage}{0.45\linewidth} \centering
    \includegraphics[width=\linewidth]{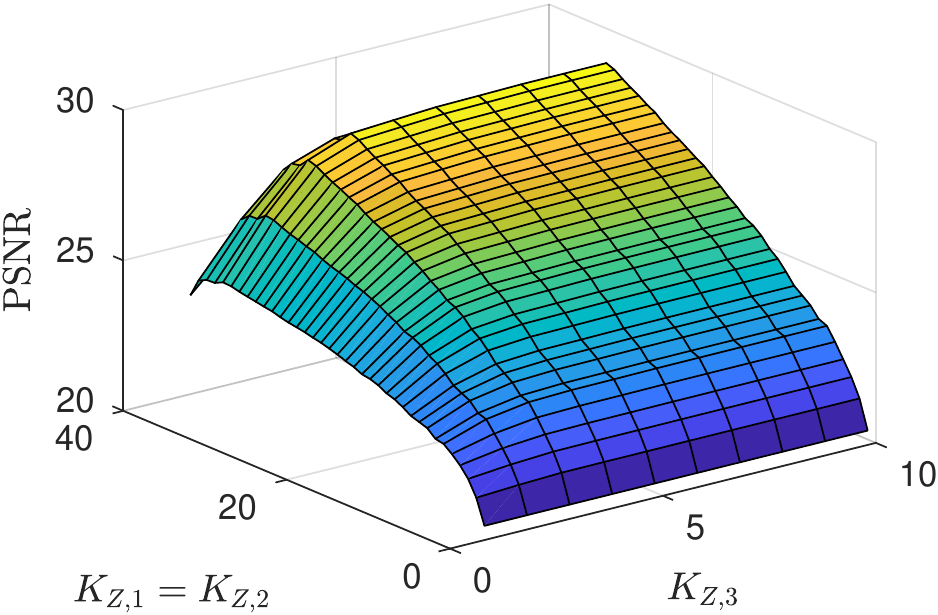}
    \end{minipage}
    \begin{minipage}{0.45\linewidth} \centering
    \includegraphics[width=\linewidth]{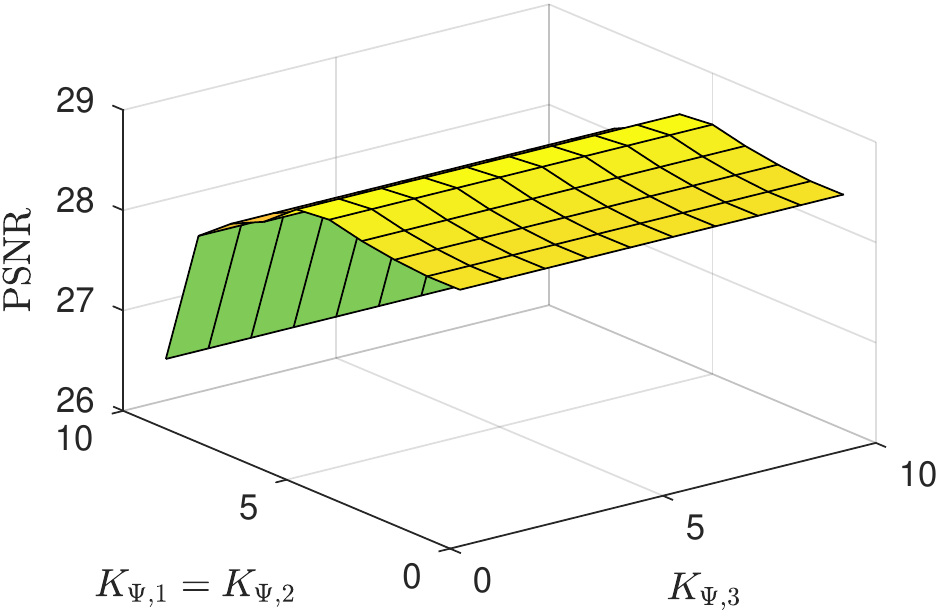}
    \end{minipage}
    \begin{minipage}{0.45\linewidth} \centering
    \includegraphics[width=\linewidth]{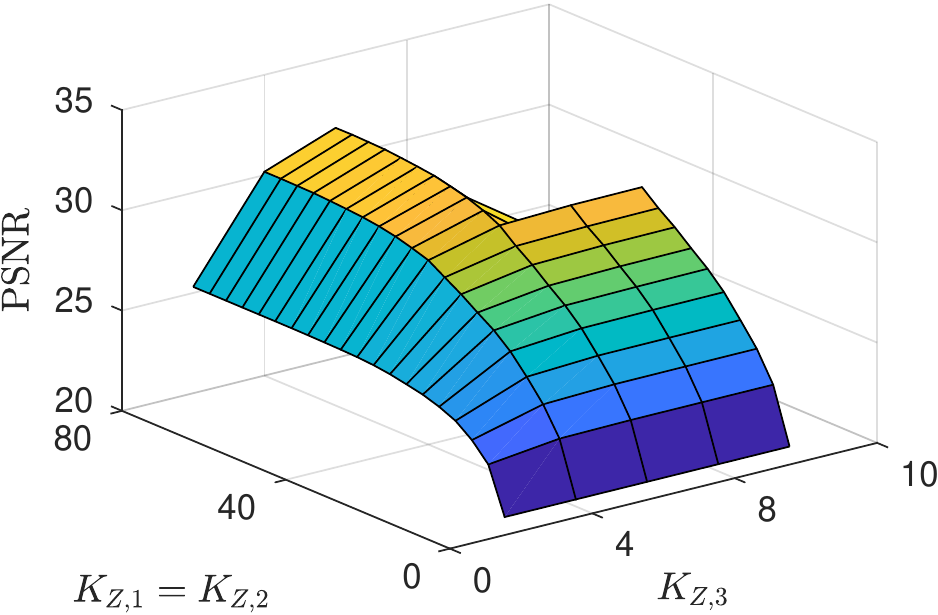}
    \end{minipage}
    \begin{minipage}{0.45\linewidth} \centering
    \includegraphics[width=\linewidth]{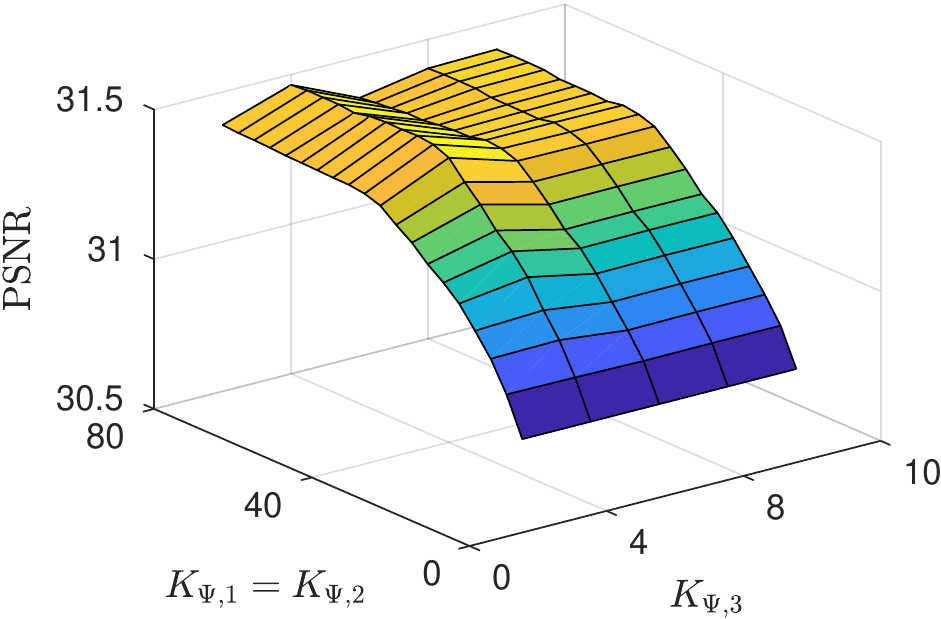}
    \end{minipage}
    \caption{\corange{Sensitivity analysis of the proposed methods Lockwood image. Top: PSNR of CT-STAR as a function of the ranks of $\tensor{Z}_h$ (left) and $\tensor{\Psi}$ (right). Bottom: PSNR of CB-STAR as a function of the ranks of $\tensor{Z}_h$ (left) and $\tensor{\Psi}$ (right)}}
    \label{fig:sensitivity_ex6}
\end{figure}

\begin{table}[ht]
\renewcommand{\arraystretch}{1.15}
\setlength\tabcolsep{3.5pt}
\centering
\caption{Ranks of the tensor-based algorithms used in the experiments}  \vspace{-5pt}
{\small 
    \begin{tabular}{c|cccc} \hline
    & CT-STAR: & CB-STAR: & SCOTT: & STEREO: \\ 
    & $K_Z,K_{\Psi}$ & $K_Z,K_{\Psi}$ & $K$ & $K$ \\ \hline
    Lockwood      & \makecell{$(30,30,8)$,\\$(3,3,2)$} & \makecell{$(70,70,5)$,\\$(40,40,3)$} & $(60,60,5)$ & $50$\\
    Lake Isabella & \makecell{$(30,30,8)$,\\$(3,3,2)$} & \makecell{$(50,50,5)$,\\$(40,40,3)$} & $(60,60,5)$ & $50$\\
    Lake Tahoe    & \makecell{$(30,30,10)$,\\$(3,3,1)$} & \makecell{$(35,35,9)$,\\$(50,50,4)$} & $(40,40,7)$ & $30$\\
    Ivanpah Playa & \makecell{$(16,16,8)$,\\$(3,3,2)$} & \makecell{$(40,40,4)$,\\$(40,40,5)$} & $(30,30,30)$ & $10$\\ \hline
\end{tabular}}
\label{tab:exp_ranks_values}
\end{table}

\subsubsection{Moderate variability}
\label{sec:res_realdata_moderate}

The first pair of images considered in this example contained $80\times80$ pixels and were acquired over the region surrounding Lake Isabella, on 2018-06-27 and on 2018-08-27. The second pair of images contained $80\times100$ pixels and was acquired near Lockwood, on 2018-08-20 and on 2018-10-19. A true color representation of the HSI and MSI for this example can be seen in Fig.~\ref{fig:ex3_HSIs_MSIs_a}.
Due to the relatively small difference between the acquisition dates of both images, the HSI and MSI look similar. However, there are slight differences between them, as seen in the overall color hue of the images and in the upper right part of the Lake Isabella HSI.
The quantitative performance metrics of all algorithms are shown in Tables~\ref{tab:ex_lockwood} and~\ref{tab:ex_isabellalake}, while the reconstructed images are presented in Figs.~\ref{fig:ex_lockwood_Z_Psi_est} and~\ref{fig:ex_isabellalake_Z_Psi_est}.

\begin{table}[ht!]	
\caption{Results - Lockwood} \vspace{-5pt}
\centering	
\renewcommand{\arraystretch}{1.15}
\setlength\tabcolsep{3.5pt}
{\small  
\begin{filecontents*}{tables/tab_lockwood.txt}
    HySure	3.3764	7.789	23.6459	0.8803	4.6327
    CNMF	2.5709	5.6351	27.6035	0.8949	8.8277
    GLPHS 	2.5701	5.3248	28.3937	0.9133	4.7375
    FuVar	2.3676	4.2921	30.5882	0.9549	217.998
    LTMR	3.470868e+00	5.006614e+00	2.916171e+01	9.394440e-01	2.621605e+01
    STEREO	3.486	5.5059	28.7197	0.9292	1.1373
    SCOTT 	2.5204	4.913	29.9346	0.9478	0.1783
    CT-STAR	2.9563	5.2457	28.3632	0.916	1.818
    CB-STAR, init=interp	2.1881	4.3458	31.4697	0.9581	18.798
    CB-STAR, init=pseudoinv	2.2227	4.3426	31.4108	0.9579	17.9794
\end{filecontents*}
\pgfplotstabletypeset[header=false,
col sep = tab,
	columns/0/.style={string type},
every head row/.style={before row={\hline}, after row=\hline},
columns/0/.style={string type, column name={Algorithm}, column type/.add={@{}@{\,}}{}},
columns/1/.style={column name={SAM}, column type/.add={@{\,}|@{\,}}{}},
columns/2/.style={column name={ERGAS}, column type/.add={@{\,}|@{\,}}{}},
columns/3/.style={column name={PSNR}, column type/.add={@{\,}|@{\,}}{}},
columns/4/.style={column name={UIQI}, column type/.add={@{\,}|@{\,}}{}},
columns/5/.style={column name={time}, fixed, precision=2, column type/.add={@{\,}|@{\,}}{}},
columns ={0,1,2,3,4,5},
every row 8 column 1/.style={postproc cell content/.style={ @cell content/.add={$\bf}{$} }},
every row 3 column 2/.style={postproc cell content/.style={ @cell content/.add={$\bf}{$} }},
every row 8 column 3/.style={postproc cell content/.style={ @cell content/.add={$\bf}{$} }},
every row 8 column 4/.style={postproc cell content/.style={ @cell content/.add={$\bf}{$} }},
every row 9 column 4/.style={postproc cell content/.style={ @cell content/.add={$\bf}{$} }},
every row 6 column 5/.style={postproc cell content/.style={ @cell content/.add={$\bf}{$} }},
]{tables/tab_lockwood.txt}}
\label{tab:ex_lockwood}
\end{table}


\begin{table}[ht!]	
\caption{Results - Isabella Lake} \vspace{-5pt}
\centering	
\renewcommand{\arraystretch}{1.15}
\setlength\tabcolsep{3.5pt}
{\small  
\begin{filecontents*}{tables/tab_isabellalake.txt}
    HySure	2.4113	8.9713	21.1945	0.7255	3.5193
    CNMF	1.8549	6.0232	26.3639	0.8417	6.7207
    GLPHS 	2.0037	4.6886	29.3327	0.9227	3.8348
    FuVar	1.7253	3.6479	32.2842	0.9707	197.9484
    LTMR	3.594838e+00	5.802096e+00	2.928604e+01	9.316169e-01	2.344282e+01
    STEREO	3.216	4.8455	30.8558	0.9454	1.0159
    SCOTT	2.3592	4.681	30.7596	0.9569	0.1567
    CT-STAR	1.8251	4.471	30.5563	0.938	1.8348
    CB-STAR, init=interp	1.3755	3.7196	34.4161	0.9763	3.248
    CB-STAR, init=pseudoinv	1.5133	3.9043	33.862	0.9734	3.2545
\end{filecontents*}
\pgfplotstabletypeset[header=false,
col sep = tab,
	columns/0/.style={string type},
every head row/.style={before row={\hline}, after row=\hline},
columns/0/.style={string type, column name={Algorithm}, column type/.add={@{}@{\,}}{}},
columns/1/.style={column name={SAM}, column type/.add={@{\,}|@{\,}}{}},
columns/2/.style={column name={ERGAS}, column type/.add={@{\,}|@{\,}}{}},
columns/3/.style={column name={PSNR}, column type/.add={@{\,}|@{\,}}{}},
columns/4/.style={column name={UIQI}, column type/.add={@{\,}|@{\,}}{}},
columns/5/.style={column name={time}, fixed, precision=2, column type/.add={@{\,}|@{\,}}{}},
columns ={0,1,2,3,4,5},
every row 8 column 1/.style={postproc cell content/.style={ @cell content/.add={$\bf}{$} }},
every row 3 column 2/.style={postproc cell content/.style={ @cell content/.add={$\bf}{$} }},
every row 8 column 3/.style={postproc cell content/.style={ @cell content/.add={$\bf}{$} }},
every row 8 column 4/.style={postproc cell content/.style={ @cell content/.add={$\bf}{$} }},
every row 6 column 5/.style={postproc cell content/.style={ @cell content/.add={$\bf}{$} }},
]{tables/tab_isabellalake.txt}}
\label{tab:ex_isabellalake}
\end{table}


The quantitative results show that CB-STAR achieved the overall best results for this example, outperforming the other methods in all metrics except in ERGAS, where it performed very similarly to FuVar (which yielded the best results for this metric). CT-STAR, on the other hand, performed similarly to STEREO and SCOTT, being limited by the more stringent constraints on the image ranks. 
The visual inspection of the results indicates that CB-STAR provides reconstructions closest to the ground truth when compared to the remaining methods. Although FuVar also provided good results, it yielded a slightly worse representation of the road in the left part of the Lockwood HSI, as well as more aberrations in the color of the light-brown regions in the middle of the Isabella Lake scene (which are not seen in the results of CB-STAR). \corange{LTMR,} STEREO and SCOTT, not being able to account for variability, yielded slight color aberrations in the reconstructions, which are most clearly seen in the central part of the Isabella Lake image, while CT-STAR produced significant artifacts due to the stringent rank constraints.
The estimated factors $\opPspec(\widehat{\tensor{\Psi}})$ are in agreement with the localized changes observed in Fig.~\ref{fig:ex3_HSIs_MSIs_a}, particularly in the upper-central area of the Isabella Lake image pair, which is subject to local illumination changes.
The computation times of the algorithms show a large difference between that of FuVar and those of the other algorithms, which indicates that CB-STAR achieves better results at a significantly smaller computational complexity.

\subsubsection{Significant variability}


The remaining image pairs used in this example were acquired over the Ivanpah Playa and over Lake Tahoe area. The Ivanpah Playa image pair contained $80\times128$ pixels and was acquired on 2015-10-26 and on 2017-12-17. For the Lake Tahoe region, we considered two different image pairs (``A'' and ``B''), both with $100\times80$ pixels, the first one acquired on 2014-10-04 and on 2017-10-24, and the second one acquired on 2014-09-19 and on 2017-10-24. A true color representation of the HSI and MSI for this example can be seen in Fig.~\ref{fig:ex3_HSIs_MSIs_b}.
Due to the considerable difference between the acquisition date/time of the HSI and MSI, significant differences can be found between them. For the Ivanpah Playa images, there are large variations between the sand colors in the central part of the image. For the Lake Tahoe region, significant differences are observed in both image pairs, with differences in the color hues of the ground and of the crop circles for the image pair A, and also a large change in the water level of the lake in the image pair B.
The quantitative performance metrics of all algorithms are shown in Tables~\ref{tab:ex_playa},~\ref{tab:ex_tahoe}, and~\ref{tab:ex_tahoeChanged}, while the reconstructed images are presented in Figs.~\ref{fig:ex_playa_Z_Psi_est},~\ref{fig:ex_tahoe_Z_Psi_est} and~\ref{fig:ex_tahoeChanged_Z_Psi_est}.

\begin{table}[ht!]	
\caption{Results - Ivanpah Playa} \vspace{-5pt}
\centering	
\renewcommand{\arraystretch}{1.15}
\setlength\tabcolsep{3.5pt}
{\small  
\begin{filecontents*}{tables/tab_playa.txt}
    HySure	1.7813	4.5332	23.3462	0.5699	6.1914
    CNMF	1.2428	3.2198	26.6513	0.7789	16.3606
    GLPHS	1.5927	3.1708	26.8431	0.8196	5.9731
    FuVar	1.0581	2.041	30.6036	0.9572	254.9709
    LTMR	3.724827e+01	1.951531e+03	1.097476e+01	4.584361e-01	3.052776e+01
    STEREO	28.1739	9.84E+03	20.4312	0.6095	0.7426
    SCOTT	35.7413	385.2807	11.4045	0.4444	0.2133
    CT-STAR	1.485	3.4361	26.0943	0.7054	0.18
    CB-STAR, init=interp	1.2208	1.8398	31.5591	0.9499	71.4719
    CB-STAR, init=pseudoinv	1.5133	2.1367	30.3019	0.9153	48.938
\end{filecontents*}
\pgfplotstabletypeset[header=false,
col sep = tab,
	columns/0/.style={string type},
every head row/.style={before row={\hline}, after row=\hline},
columns/0/.style={string type, column name={Algorithm}, column type/.add={@{}@{\,}}{}},
columns/1/.style={column name={SAM}, column type/.add={@{\,}|@{\,}}{}},
columns/2/.style={column name={ERGAS}, column type/.add={@{\,}|@{\,}}{}},
columns/3/.style={column name={PSNR}, column type/.add={@{\,}|@{\,}}{}},
columns/4/.style={column name={UIQI}, column type/.add={@{\,}|@{\,}}{}},
columns/5/.style={column name={time}, fixed, precision=2, column type/.add={@{\,}|@{\,}}{}},
columns ={0,1,2,3,4,5},
every row 3 column 1/.style={postproc cell content/.style={ @cell content/.add={$\bf}{$} }},
every row 8 column 2/.style={postproc cell content/.style={ @cell content/.add={$\bf}{$} }},
every row 8 column 3/.style={postproc cell content/.style={ @cell content/.add={$\bf}{$} }},
every row 3 column 4/.style={postproc cell content/.style={ @cell content/.add={$\bf}{$} }},
every row 7 column 5/.style={postproc cell content/.style={ @cell content/.add={$\bf}{$} }},
]{tables/tab_playa.txt}}
\label{tab:ex_playa}
\end{table}

\begin{table}[ht!]	
\caption{Results - Lake Tahoe A} \vspace{-5pt}
\centering	
\renewcommand{\arraystretch}{1.15}
\setlength\tabcolsep{3.5pt}
{\small  
\begin{filecontents*}{tables/tab_tahoe.txt}
    HySure	11.2982	13.9942	17.3675	0.7125	4.4982
    CNMF	8.7865	14.5934	18.3694	0.709	12.1034
    GLPHS 	5.6479	7.4541	24.0828	0.9121	4.6505
    FuVar	3.9147	4.7313	27.984	0.9749	270.9105
    LTMR	3.445439e+01	1.357418e+03	1.379577e+01	5.221674e-01	2.494082e+01
    STEREO	27.0675	1.54E+03	20.1943	0.6793	0.9178
    SCOTT	33.1684	4.31E+04	11.2121	0.3889	1.4739
    CT-STAR	5.4107	5.2514	27.2491	0.9584	2.8773
    CB-STAR, init=interp	4.2479	3.7836	30.1047	0.9773	63.7065
    CB-STAR, init=pseudoinv	4.6957	3.9411	29.669	0.9778	31.6002
\end{filecontents*}
\pgfplotstabletypeset[header=false,
col sep = tab,
	columns/0/.style={string type},
every head row/.style={before row={\hline}, after row=\hline},
columns/0/.style={string type, column name={Algorithm}, column type/.add={@{}@{\,}}{}},
columns/1/.style={column name={SAM}, column type/.add={@{\,}|@{\,}}{}},
columns/2/.style={column name={ERGAS}, column type/.add={@{\,}|@{\,}}{}},
columns/3/.style={column name={PSNR}, column type/.add={@{\,}|@{\,}}{}},
columns/4/.style={column name={UIQI}, column type/.add={@{\,}|@{\,}}{}},
columns/5/.style={column name={time}, fixed, precision=2, column type/.add={@{\,}|@{\,}}{}},
columns ={0,1,2,3,4,5},
every row 3 column 1/.style={postproc cell content/.style={ @cell content/.add={$\bf}{$} }},
every row 8 column 2/.style={postproc cell content/.style={ @cell content/.add={$\bf}{$} }},
every row 8 column 3/.style={postproc cell content/.style={ @cell content/.add={$\bf}{$} }},
every row 8 column 4/.style={postproc cell content/.style={ @cell content/.add={$\bf}{$} }},
every row 9 column 4/.style={postproc cell content/.style={ @cell content/.add={$\bf}{$} }},
every row 5 column 5/.style={postproc cell content/.style={ @cell content/.add={$\bf}{$} }},
]{tables/tab_tahoe.txt}}
\label{tab:ex_tahoe}
\end{table}

\begin{table}[ht!]	
\caption{Results - Lake Tahoe B} \vspace{-5pt}
\centering	
\renewcommand{\arraystretch}{1.15}
\setlength\tabcolsep{3.5pt}
{\small  
\begin{filecontents*}{tables/tab_tahoeChanged.txt}
    HySure	7.1707	19.0755	13.6207	0.3528	4.3599
    CNMF	8.0828	14.7049	16.1632	0.4183	12.4197
    GLPHS 	3.6072	5.5821	24.5744	0.8563	4.531
    FuVar	2.5825	3.3822	28.8579	0.9557	342.393
    LTMR	3.807682e+01	1.206541e+03	1.202570e+01	4.589613e-01	2.445631e+01
    STEREO	28.1818	6.22E+03	19.9861	0.629	0.7518
    SCOTT	38.4464	2.96E+03	10.874	0.3121	1.4171
    CT-STAR	3.0745	4.2989	26.8158	0.9216	2.882
    CB-STAR, init=interp	2.1709	2.6358	31.19	0.967	46.4645
    CB-STAR, init=pseudoinv	2.3364	2.7314	30.7424	0.9647	30.075
\end{filecontents*}
\pgfplotstabletypeset[header=false,
col sep = tab,
	columns/0/.style={string type},
every head row/.style={before row={\hline}, after row=\hline},
columns/0/.style={string type, column name={Algorithm}, column type/.add={@{}@{\,}}{}},
columns/1/.style={column name={SAM}, column type/.add={@{\,}|@{\,}}{}},
columns/2/.style={column name={ERGAS}, column type/.add={@{\,}|@{\,}}{}},
columns/3/.style={column name={PSNR}, column type/.add={@{\,}|@{\,}}{}},
columns/4/.style={column name={UIQI}, column type/.add={@{\,}|@{\,}}{}},
columns/5/.style={column name={time}, fixed, precision=2, column type/.add={@{\,}|@{\,}}{}},
columns ={0,1,2,3,4,5},
every row 8 column 1/.style={postproc cell content/.style={ @cell content/.add={$\bf}{$} }},
every row 8 column 2/.style={postproc cell content/.style={ @cell content/.add={$\bf}{$} }},
every row 8 column 3/.style={postproc cell content/.style={ @cell content/.add={$\bf}{$} }},
every row 8 column 4/.style={postproc cell content/.style={ @cell content/.add={$\bf}{$} }},
every row 5 column 5/.style={postproc cell content/.style={ @cell content/.add={$\bf}{$} }},
]{tables/tab_tahoeChanged.txt}}
\label{tab:ex_tahoeChanged}
\end{table}


The quantitative results show that \mbox{CB-STAR} achieved again the overall best results for this example, outperforming the remaining algorithms in most metrics, except in the SAM and UIQI for the Ivanpah Playa HRI and in the SAM of the Lake Tahoe A HRI. Moreover, there was a stronger gap between the performance of the methods that consider variability and the remaining algorithms. \mbox{CT-STAR}, although better than \corange{LTMR,} STEREO and SCOTT, performed significantly worse than CB-STAR due to its stringent constraints on the image ranks. 
The visual inspection of the results again indicates that CB-STAR provides reconstructions closest to the ground truth when compared to the remaining methods. Although FuVar also provided good results \corange{(as it accounts for spectral variability)}, the reconstructions by CB-STAR were closer to the ground truth, as can be observed in the color shades of the upper part of the Ivanpah Playa image and of the crop circles of the Lake Tahoe~A image, and especially in the overall colors in the more uniform regions containing soil and water and vegetation in the Lake Tahoe B image. \corange{However, FuVar results showed slightly sharper edges in some regions (e.g., around the crop circles in the Lake Tahoe images), which occurs due to the use of a Total Variation spatial regularization. Nonetheless, a spatial regularization can  also be incorporated to the CB-STAR cost function in~\eqref{eq:opt_prob_fus_btd} to achieve a similar effect.}
\corange{LTMR,} STEREO and SCOTT, on the other hand, produced significant artifacts in all reconstructed images. CT-STAR also produced significant artifacts, which as in the previous example are due to the stringent rank constraints\corange{, with reconstructions visually worse than those of CB-STAR and FuVar}.
The estimated factors $\opPspec(\widehat{\tensor{\Psi}})$ were in close agreement with the variability seen in the scenes, notably in the sand region of the central part of the Ivanpah Playa image (which lies at the botton of a hill), and in the regions near the lake in the Lake Tahoe A and B images, which undergo variations in the water level. Moreover, the overall amplitude of the variables was significantly larger than in the previous example, in which the differences between the images were more moderate.
The computation times of all algorithms were similar to those observed in the previous example, except for that of CB-STAR, which was higher since it underwent a larger number of iterations for the data in this example. Nonetheless, the computation times of CB-STAR were still considerably smaller than those of FuVar.

\begin{figure*}
    \centering
    \resizebox{0.85\linewidth}{!}
    {
    \includegraphics[height=10ex]{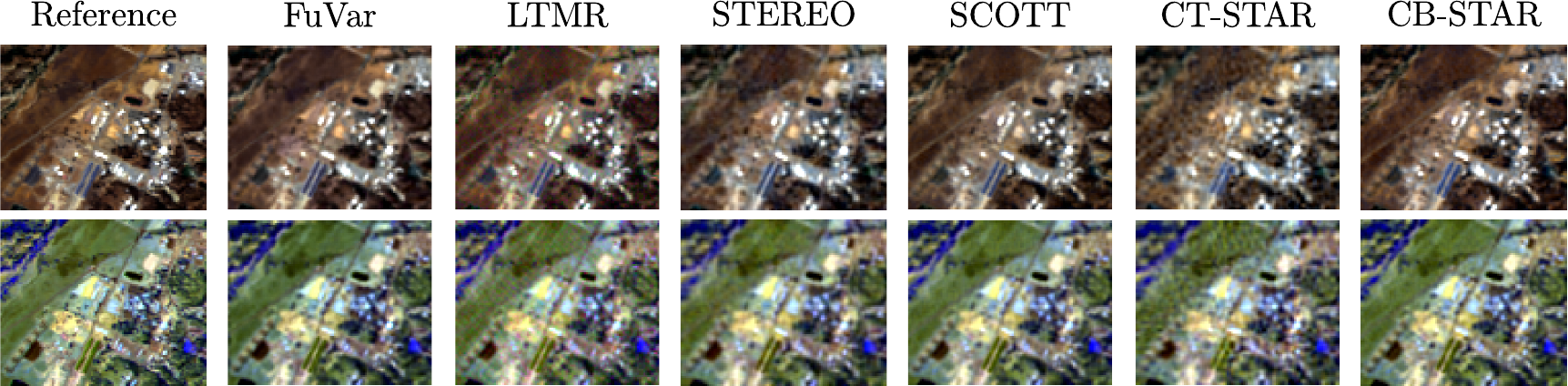}
    \includegraphics[height=9.85ex]{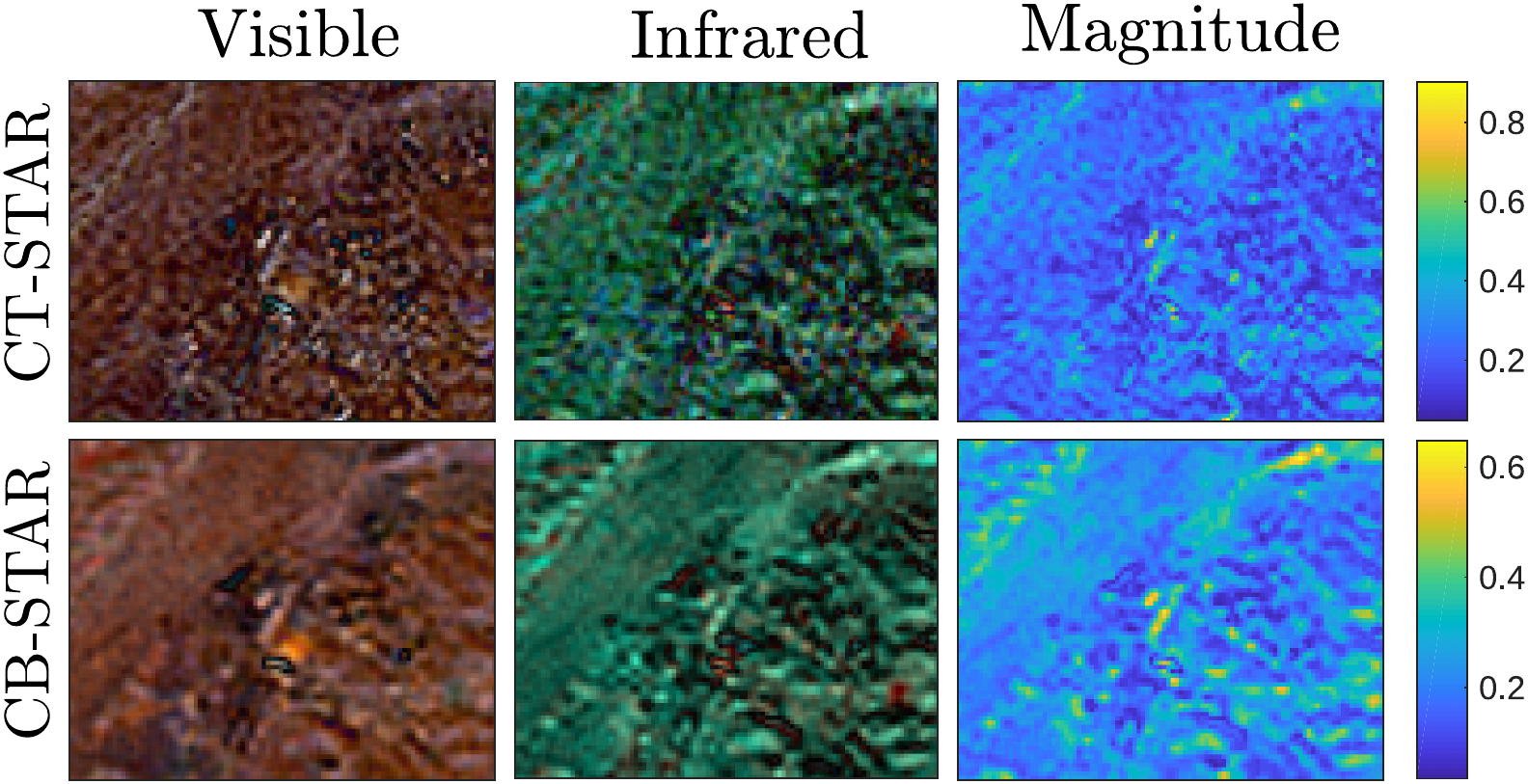}
    }
    \vspace{-0.3cm}
    \caption{Left: Visible (top) and infrared (bottom) representation of the true and estimated versions of the Lockwood HSI. Right: Spectrally degraded additive scaling factors $\opPspec(\tensor{\Psi})$ estimated by CT-STAR and CB-STAR.}
    \label{fig:ex_lockwood_Z_Psi_est}
\end{figure*}

\begin{figure*}
    \centering
    \resizebox{0.85\linewidth}{!}
    {
    \includegraphics[height=10ex]{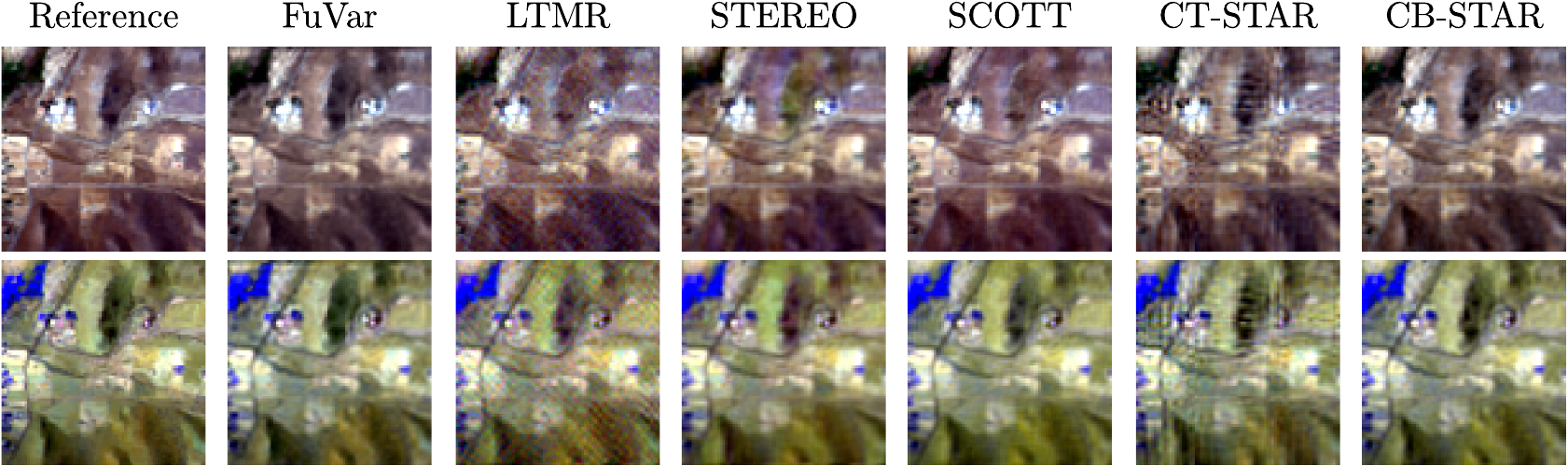}
    \includegraphics[height=9.9ex]{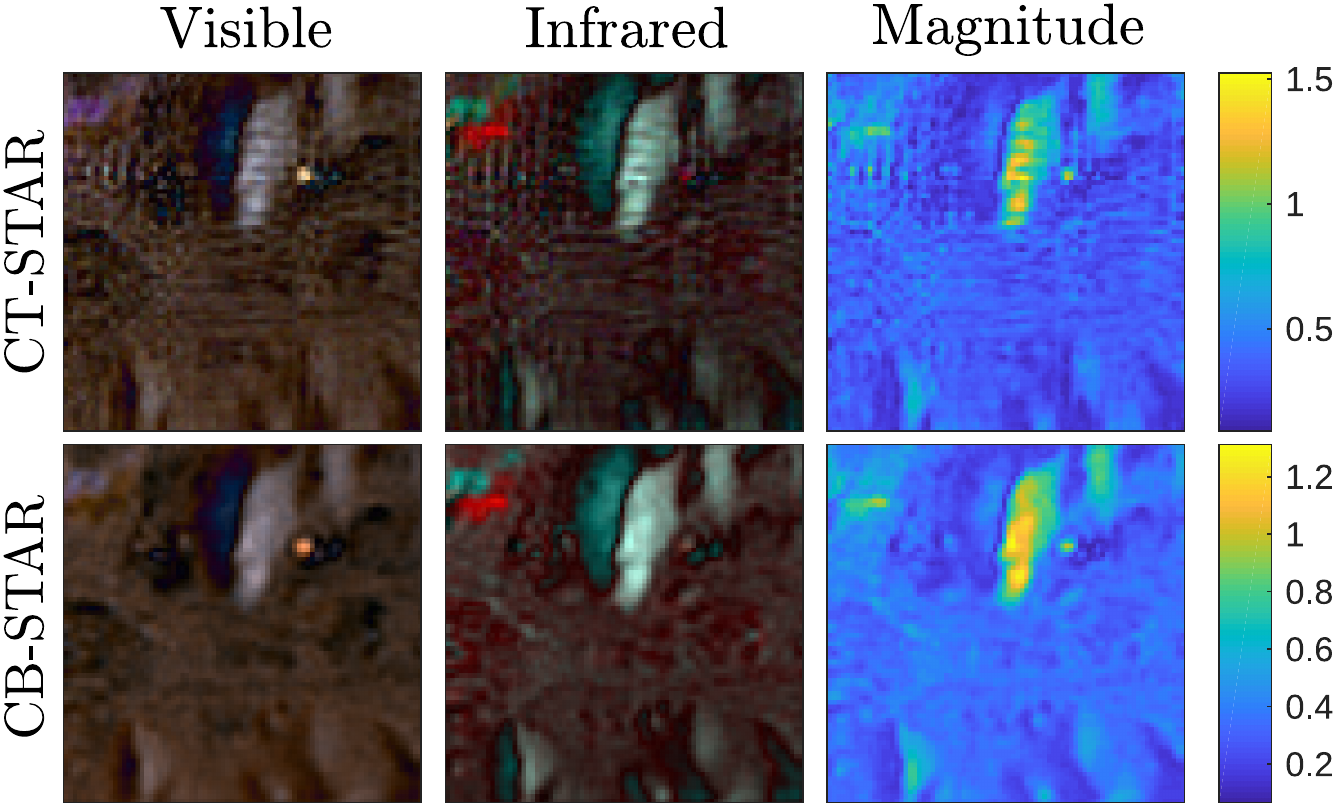}
    }
    \vspace{-0.3cm}
    \caption{Left: Visible (top) and infrared (bottom) representation of the true and estimated versions of the Isabella Lake HSI. Right: Spectrally degraded additive scaling factors $\opPspec(\tensor{\Psi})$ estimated by CT-STAR and CB-STAR.}
    \label{fig:ex_isabellalake_Z_Psi_est}
\end{figure*}

\begin{figure*}
    \centering
    \resizebox{0.85\linewidth}{!}
    {
    \includegraphics[height=10ex]{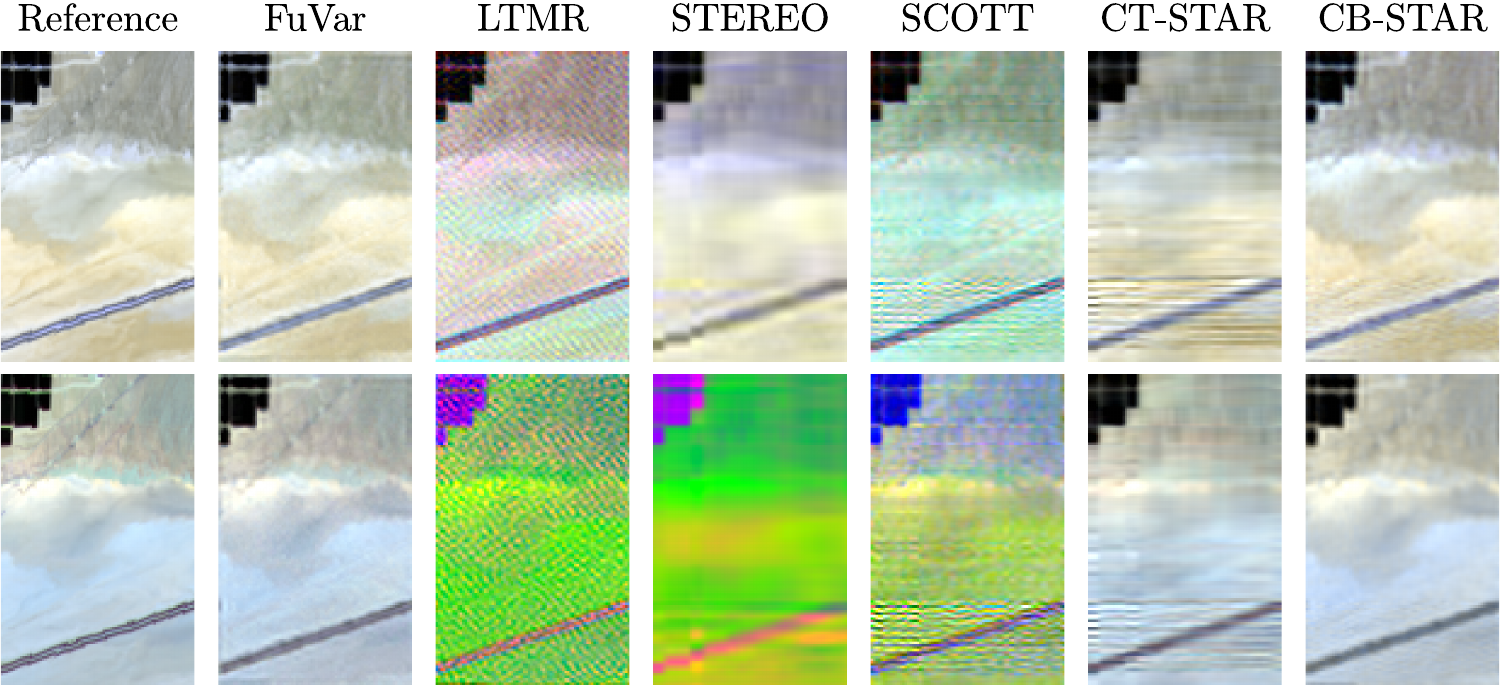}
    \includegraphics[height=9.95ex]{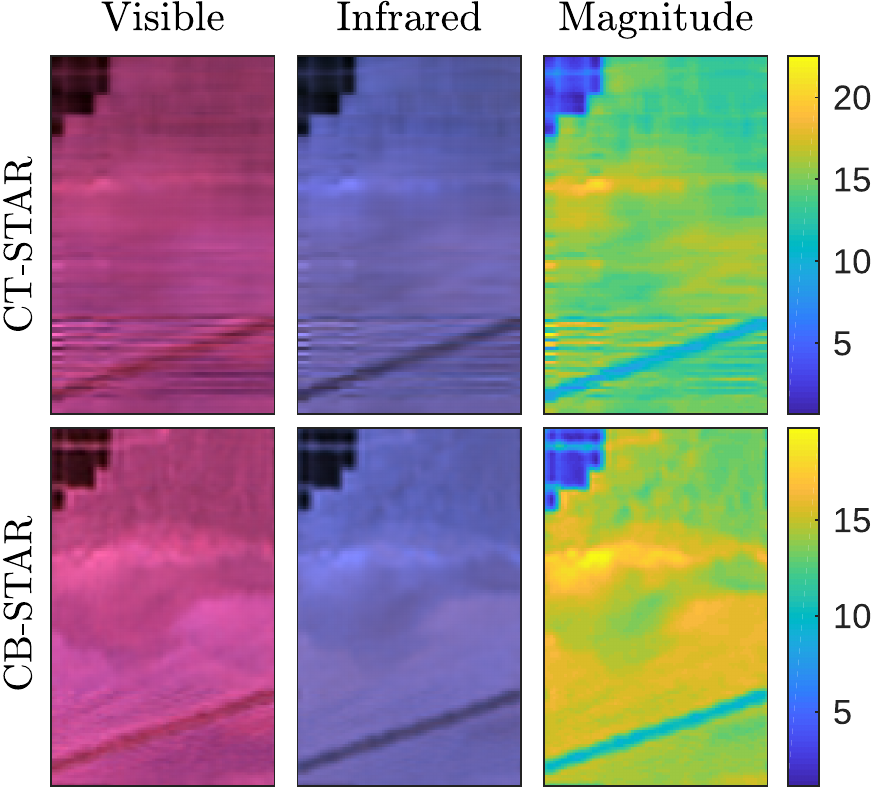}
    }
    \vspace{-0.3cm}
    \caption{Left: Visible (top) and infrared (bottom) representation of the true and estimated versions of the Ivanpah Playa HSI. Right: Spectrally degraded additive scaling factors $\opPspec(\tensor{\Psi})$ estimated by CT-STAR and CB-STAR.}
    \label{fig:ex_playa_Z_Psi_est}
\end{figure*}

\begin{figure*}
    \centering
    \resizebox{0.85\linewidth}{!}
    {
    \includegraphics[height=10ex]{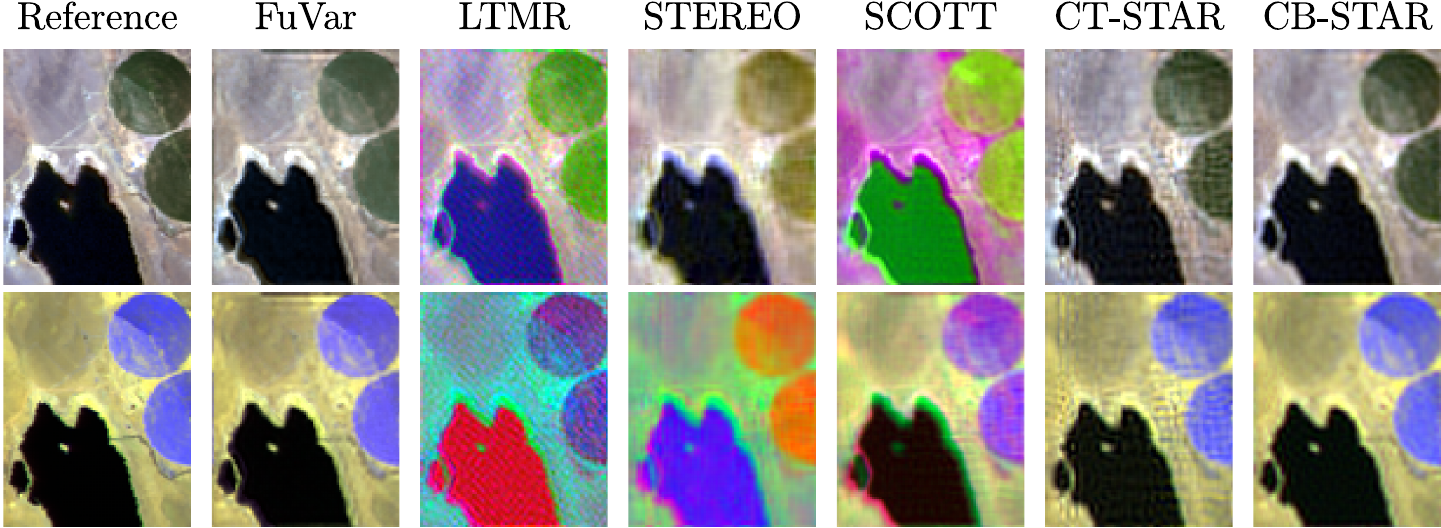}
    \includegraphics[height=9.9ex]{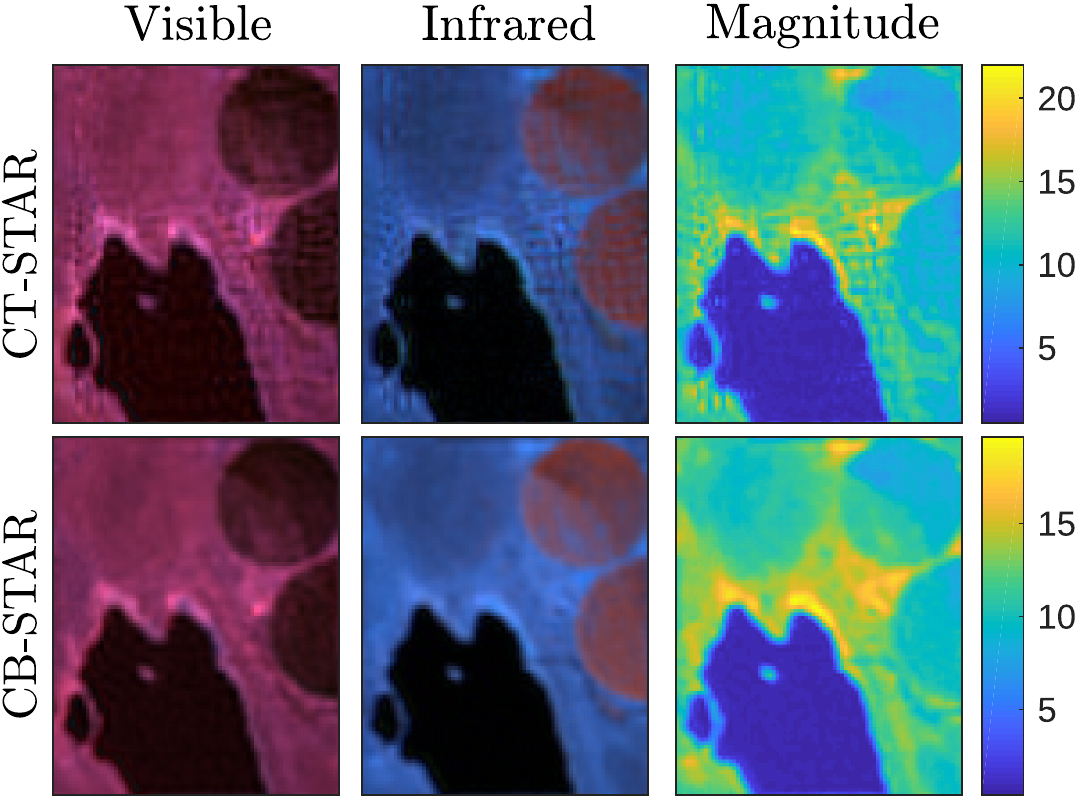}
    }
    \vspace{-0.3cm}
    \caption{Left: Visible (top) and infrared (bottom) representation of the true and estimated versions of the Lake Tahoe A HSI. Right: Spectrally degraded additive scaling factors $\opPspec(\tensor{\Psi})$ estimated by CT-STAR and CB-STAR.}
    \label{fig:ex_tahoe_Z_Psi_est}
\end{figure*}

\begin{figure*}[thb]
    \centering
    \resizebox{0.85\linewidth}{!}
    {
    \includegraphics[height=10ex]{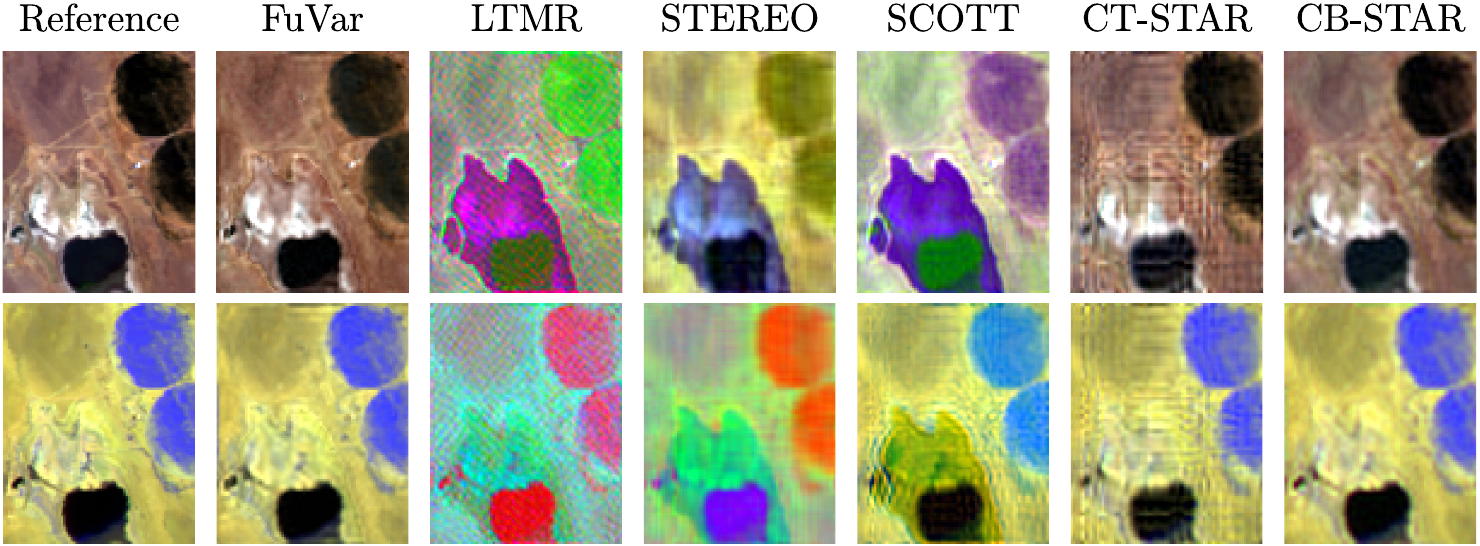}
    \includegraphics[height=9.9ex]{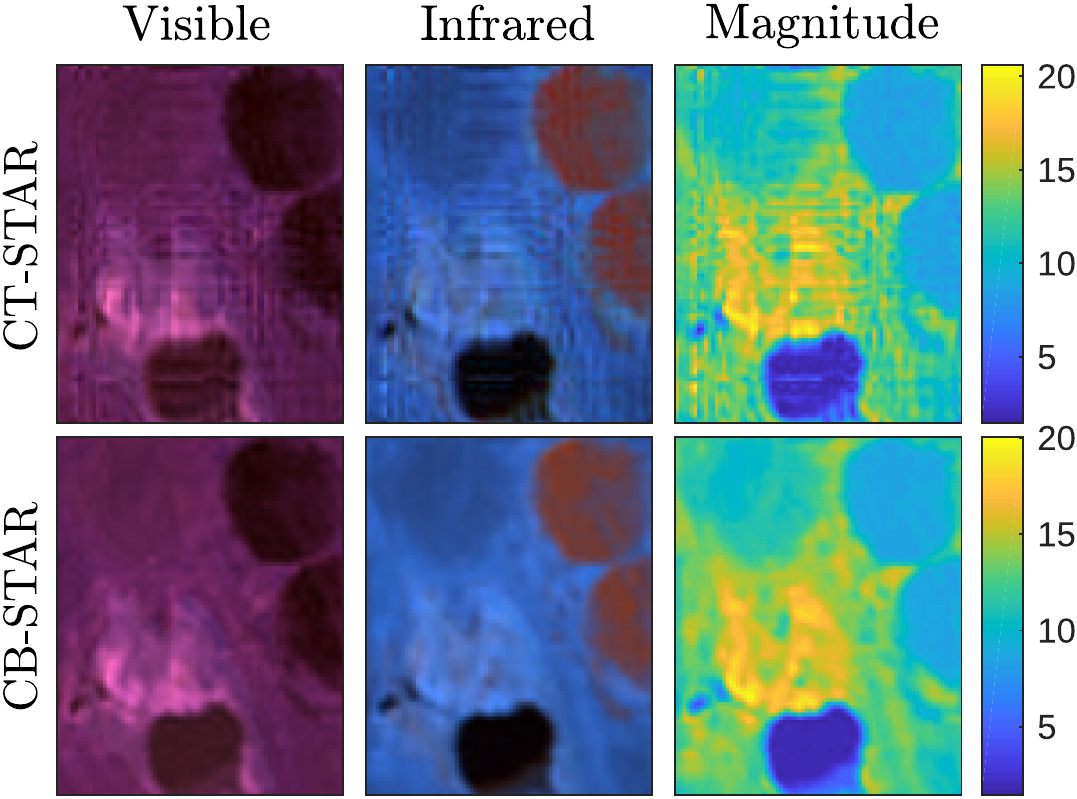}
    }
    \vspace{-0.3cm}
    \caption{Left: Visible (top) and infrared (bottom) representation of the true and estimated versions of the Lake Tahoe B HSI. Right: Spectrally degraded additive scaling factors $\opPspec(\tensor{\Psi})$ estimated by CT-STAR and CB-STAR.}
    \label{fig:ex_tahoeChanged_Z_Psi_est}
\end{figure*}

\section{Conclusions}
\label{sec:conclusions}

In this paper, we proposed a novel framework for multimodal (hyperspectral and multispectral) image fusion accounting for spatially and spectrally localized changes. We first studied the general identifiability of the considered model, which becomes more ambiguous due to the presence of changes. Then, assuming that the high resolution image and the variation factors admit a Tucker decomposition, two new algorithms were proposed -- one being purely algebraic (which was computationally more efficient), and another based on an optimization procedure (which allowed for more relaxed specification of the multilinear ranks). Theoretical guarantees for the exact recovery of the high resolution image were provided for both algorithms. The proposed optimization-based algorithm achieved superior experimental performance in the presence of spectral and spatial variations between the images, while also exhibiting a smaller computational cost.

\appendices

\section{Optimization of cost function~\eqref{eq:opt_subp_Z}}
\label{sec:appendix1}

\corange{To solve~\eqref{eq:opt_subp_Z}, we consider a block coordinate descent strategy (as in, e.g.,~\cite{lathauwer2008tensor_BTD3_algorithms}), which successively minimizes the cost function in~\eqref{eq:opt_subp_Z} with respect to each of the variables $\tensor{G}_Z$, $\bB_{Z,i}$, $i\in\{1,2,3\}$,} while keeping the remaining ones fixed.
Let us denote by $J\big(\bX\,|\,\bTheta_{\setminus\{\bX\}}\big)$ the cost function $J$ in which all variables but $\bX$ are fixed.
Note that we apply the QR factorization after computing each of the factor matrices $\bB_{\iota,j}$ to constrain them to be unitary at all the iterations,  as performed in~\cite{lathauwer2008tensor_BTD3_algorithms}. This normalization prevents convergence issues by avoiding under/over-flow and keeping these matrices well-conditioned. \corange{Starting from an initialization $\tensor{G}_Z$, $\bB_{Z,i}$, $i\in\{1,2,3\}$, this procedure is repeated for $F$ inner iterations.}

\begin{algorithm} [thb]
\footnotesize
\SetKwInOut{Input}{Input}
\SetKwInOut{Output}{Output}
\caption{\mbox{Block coordinate descent to solve~\eqref{eq:opt_subp_Z}}
\label{alg:alg3}}
\Input{\mbox{Tensors $\tensor{Y}_h$, $\tensor{X}_0$, $\tensor{G}_Z$ matrices $\bB_{Z,i}$,$i\in\{1,2,3\}$\corange{, iterations~$F$}}}
\Output{$\tensor{G}_Z$ matrices $\bB_{Z,i}$,$i\in\{1,2,3\}$}

\For{$i=1,\ldots,F$}{
Minimize $J(\tensor{G}_Z|\bTheta_{\setminus\{\tensor{G}_Z\}}^{(i)})$ w.r.t. $\tensor{G}_Z$\;
Minimize $J(\bB_{Z,1}|\bTheta_{\setminus\{\bB_{Z,1}\}}^{(i)})$ w.r.t. $\bB_{Z,1}$\;
Minimize $J(\bB_{Z,2}|\bTheta_{\setminus\{\bB_{Z,2}\}}^{(i)})$ w.r.t. $\bB_{Z,2}$\;
Minimize $J(\bB_{Z,3}|\bTheta_{\setminus\{\bB_{Z,3}\}}^{(i)})$ w.r.t. $\bB_{Z,3}$\;}
\end{algorithm}

\subsubsection{Optimizing w.r.t. $\bB_{Z,i}$, $i\in\{1,2,3\}$}

To save space, we present only the case where $i=1$. The extension to $i\in\{2,3\}$ is straightforward.
Note that the cost function $J\big(\bB_{Z,1}\,|\,\bTheta_{\setminus\{\bB_{Z,1}\}}\big)$ can be equivalently reformulated using the mode-$1$ matricization as
\begin{align} \label{eq:opt_subp_1_i}
    J\big(& \bB_{Z,1}\,|\,\bTheta_{\setminus\{\bB_{Z,1}\}}\big) = \big\|\unfold{\bY_{\!h}}{1}^\top - \bX_2 \bB_{Z,1}^\top\bP_1^\top\big\|_F^2
    \nonumber\\ &
    + \corange{\lambda} \big\|\unfold{\bY_{\!m}}{1}^\top - \bX_3 - \bX_1 \bB_{Z,1}^\top\big\|_F^2 \,,
\end{align}
where matrices $\bX_1$, $\bX_2$ and $\bX_3$ are given by
\begin{align}
    \bX_1 &= \big({\bB}_{Z,3} \otimes \bP_3{\bB}_{Z,2}\big) \unfold{\bG_Z}{1}^\top \,,
    \\
    \bX_2 &= \big(\bP_2{\bB}_{Z,3} \otimes {\bB}_{Z,2}\big) \unfold{\bG_Z}{1}^\top \,,
    \\
    \bX_3 & = \big( \bP_3 {\bB}_{\Psi,3} \otimes {\bB}_{\Psi,2}\big)\unfold{\bG_{\Psi}}{1}^\top {\bB}_{\Psi,1}^\top \,.
\end{align}

Computing the derivative of~\eqref{eq:opt_subp_1_i} and setting it equal to zero results in the following expression:
\begin{align}
\begin{split}
    \corange{\lambda}\bX_1^\top\bX_1 {\bB}_{Z,1}^\top 
    & + \bX_2^\top\bX_2 {\bB}_{Z,1}^\top\bP_1^\top\bP_1
    \\
    = \corange{\lambda} & \bX_1^\top(\unfold{\bY_{\!m}}{1}^\top - \bX_3)
    + \bX_2^\top\unfold{\bY_{\!h}}{1}^\top\bP_1 \,.
\end{split}
\end{align}
This is a Sylvester equation that can be directly solved using existing software with, e.g., the Hessenberg-Schur or the Bartels-Stewart algorithms (see~\cite{simoncini2016solversSylvesterEq} and references therein).

\subsubsection{Optimizing w.r.t. $\tensor{G}_{Z}$}

Cost function $J\big(\tensor{G}_{Z}\,|\,\bTheta_{\setminus\{\tensor{G}_{Z}\}}\big)$ can be equivalently reformulated using the tensor vectorization as
\begin{align} \label{eq:opt_subp_4_i}
    J\big(\tensor{G}_{Z}\,|\,\bTheta_{\setminus\{\tensor{G}_{Z}\}}\big) 
    ={} & \big\|\vect(\tensor{Y}_h) - \bX_2 \,\bg_Z \,\big\|_F^2
    \nonumber \\ &
    + \corange{\lambda} \big\|\vect(\tensor{Y}_m) - \bx_3 - \bX_1 \, \bg_Z \, \big\|_F^2 \,,
\end{align}
where $\bg_Z=\vect(\tensor{G}_{Z})$ is the vectorization of the core tensor, and $\bX_1$, $\bX_2$ and $\bx_3$ are given by
\begin{align}
    \bX_1 & = \big(\bP_3 \bB_{Z,3} \otimes \bB_{Z,2} \otimes \bB_{Z,1} \big) \,,
    \\
    \bX_2 & = \big(\bB_{Z,3} \otimes \bP_2 \bB_{Z,2} \otimes \bP_1 \bB_{Z,1} \big) \,,
    \\
    \bx_3 & = \big(\bP_3\bB_{\Psi,3} \otimes \bB_{\Psi,2} \otimes \bB_{\Psi,1} \big) \vect(\tensor{G}_{\Psi}) \,.
\end{align}
The solution that minimizes~\eqref{eq:opt_subp_4_i} can be computed through the normal equations, which can be written as
\begin{align}
    (\corange{\lambda}\bX_1^\top\bX_1 + \bX_2^\top\bX_2) \bg_Z {}={} &  \corange{\lambda}\bX_1^\top(\vect(\tensor{Y}_m)-\bx_3)
    \nonumber \\
    & + \bX_2^\top\vect(\tensor{Y}_h) \,.
\end{align}
As shown in~\cite{prevost2020coupledTucker_hyperspectralSRR_TSP}, this set of equations can be alternatively interpreted as a generalized Sylvester equation, for which efficient solvers can be used.

\bibliographystyle{IEEEtran}
\bibliography{references_fus2,references_VarRevPaper,references_fus_old}

\end{document}

%% file: MathSymbolDefs.tex
\usepackage[english]{babel}
\usepackage[utf8x]{inputenc}
\usepackage[T1]{fontenc}

\usepackage{amsmath,amssymb,mathrsfs,bbm}
\usepackage{graphicx}
\usepackage[colorinlistoftodos]{todonotes}
\usepackage[colorlinks=true, allcolors=blue]{hyperref}
\usepackage{multirow}
\usepackage{makecell}
\usepackage[lined,linesnumbered,ruled]{algorithm2e}

\usepackage{psfrag,epsfig,graphics}
\usepackage{amsmath,amsthm,amssymb,multirow}
\usepackage{mathbbol}
\usepackage{amssymb}   
\usepackage{url}  
\usepackage{xstring}
\usepackage{cite}

\usepackage{mathabx} 

\DeclareSymbolFontAlphabet{\amsmathbb}{AMSb}%

\newcommand{\cp}[1]{\ifmmode {\mathcal{#1}}\else ${\mathcal{#1}}$\fi}
	
\newcommand{\bA}{\boldsymbol{A}}
\newcommand{\bB}{\boldsymbol{B}}
\newcommand{\bC}{\boldsymbol{C}}

\newcommand{\bG}{\boldsymbol{G}}

\newcommand{\bP}{\boldsymbol{P}}
\newcommand{\bQ}{\boldsymbol{Q}}

\newcommand{\bS}{\boldsymbol{S}}

\newcommand{\bT}{\boldsymbol{T}}

\newcommand{\bX}{\boldsymbol{X}}
\newcommand{\bY}{\boldsymbol{Y}}

\newcommand{\bg}{\boldsymbol{g}}

\newcommand{\bx}{\boldsymbol{x}}

\newcommand{\bTheta}{\boldsymbol{\Theta}}

\newcommand{\cb}[1]{\boldsymbol{#1}}

\newcommand\tensor[1]{%
  \ifcat\noexpand#1\relax 
    \mathbb{#1}
  \else
      \if\relax\detokenize\expandafter{\romannumeral-0#1}\relax  
        \mathbb{#1}
      \else
        \mathcal{#1}
      \fi
  \fi }

\newcommand{\unfold}[2]{{#1}_{({#2})}}

\newcommand{\vect}{\operatorname{vec}}
\newcommand{\rank}{\operatorname{rank}}
\newcommand{\vspan}{\operatorname{span}}
\newcommand{\tSVD}{\operatorname{tSVD}}
\newcommand{\Ex}{\operatorname{E}}

\newcommand{\opPspat}{\mathscr{P}_{1,2}}
\newcommand{\opPspec}{\mathscr{P}_{3}}

\usepackage{color}  

\def\corange{\textcolor{orange}}

\def\cmag{\textcolor{magenta}}
\definecolor{darkgreen}{rgb}{0., 0.4, 0.}
\newcommand{\cdgreen}{\textcolor{darkgreen}}

\newtheorem{definition}{Definition}

\newtheorem{theorem}{Theorem}